\documentclass[
  reprint,
  amsmath,
  amssymb,
  aps,
  pre,
  floatfix
]{revtex4-2}
\usepackage{graphicx,amsthm}
\usepackage{dcolumn}
\usepackage{bm}
\usepackage{xcolor}
\usepackage{comment}

\newtheorem{theorem}{Theorem}[section]
\newtheorem{lemma}[theorem]{Lemma}
\newtheorem{proposition}[theorem]{Proposition}
\newtheorem{corollary}[theorem]{Corollary}
\newtheorem{result}{Result}

\theoremstyle{definition}
\newtheorem{definition}[theorem]{Definition}
\newtheorem{remark}[theorem]{Remark}

\newcommand{\Var}{\mathrm{Var}}
\newcommand{\Cov}{\mathrm{Cov}}

\newcommand{\E}{\mathbb{E}}
\newcommand{\PP}{\mathbb{P}}

\newcommand{\bc}{\beta_n^{(L)}}
\newcommand{\Phic}{\Phi_\alpha^{(L)}}
\newcommand{\hatphi}{\widehat{\varphi}_{\alpha,L}}
\newcommand{\ZnL}{Z_n^{(L)}}
\newcommand{\GnL}{G_n^{(L)}}
\newcommand{\HnL}{H_n^{(L)}}
\newcommand{\LnL}{\mathcal{L}_n^{(L)}}

\begin{document}

\preprint{APS/123-QED}

\title{Phenotypic Assortment and the Evolution of Cooperation in Finite Periodic Phenotype Spaces
}

\author{Dhaker Kroumi}
\email{dhaker.kroumi@kfupm.edu.sa}
 \affiliation{Department of Mathematics, King Fahd University of Petroleum and Minerals, Dhahran, S.A.}

\date{\today}


\begin{abstract}
We study the evolution of cooperation in a finite haploid population whose individuals carry both a strategy and a phenotype on a finite periodic space. Cooperators help with a probability that decays exponentially with phenotypic distance, generating graded phenotype-dependent assortment. Under weak selection and a large-population mutation scaling, we combine coalescent arguments with the spectral representation of a random walk on the discrete torus to derive an explicit benefit-to-cost threshold for cooperation to be favored in stationary abundance. Stronger phenotypic discrimination and higher phenotype-space dimension strictly lower this threshold, whereas strategy mutation raises it. In contrast, phenotype mutation has a nonmonotone effect: the threshold diverges for both very rare and very rapid phenotype mutation and therefore attains at least one minimum at an intermediate rate. The high-mutation inhibition results from mixing on the finite phenotype space and distinguishes the periodic model from its unbounded-lattice counterpart.
\end{abstract}

\maketitle




\section{Introduction}
In the Prisoner's Dilemma, a cooperator pays a cost \(c>0\) to give a benefit \(b>c\) to another individual, whereas a defector pays no cost and gives no benefit. In a well-mixed population with one-shot interactions, this creates a basic difficulty for cooperation: defectors may receive benefits without paying the cost of helping \citep{Nowak2006}.  Then, cooperation can be favored only when helpers are, in some way, more likely to interact with other helpers than with defectors. This non-random association is usually described as positive assortment.  Several familiar mechanisms can generate such assortment, including kinship \citep{Hamilton1963,Hamilton1964a,Hamilton1964b}, repeated interaction \citep{Trivers1971,Axelrod1981,Axelrod1984}, and spatial or population structure \citep{Nowak1992,Ohtsuki2006}. Although these mechanisms differ biologically, they all increase the chance that the benefit of cooperation is directed toward individuals associated with cooperation rather than toward defectors.

Population structure provides one of the most studied sources of assortment. In spatial and graph-structured populations, local interactions and replacement events can create correlations among neighboring strategies and promote cooperation under suitable conditions \citep{Nowak1992,Lieberman2005,SantosPacheco2005,Ohtsuki2006,Taylor2007}. More generally, the effect of structure depends on the update rule, graph topology, and payoff accumulation \citep{SzaboFath2007,Roca2009,PercSzolnoki2010}. Related effects arise in subdivided populations, where limited dispersal and group formation generate genetic and strategic correlations and connect assortment with kin selection and local competition \citep{KimuraWeiss1964,Rousset2004,LehmannRousset2010,Ohtsuki2010,Lessard2011a,Ohtsuki2014,Kristensen2022,MartinLessard2023,Kristensen2025}. Although these models differ in their demographic assumptions, they share the central principle that population structure changes the pattern of association among individuals and thereby the assortment available to cooperation.

Phenotypic similarity is another example of such assortment. If individuals carry heritable markers and condition their behavior on marker similarity, then cooperation can become associated with recognizable phenotypic types. This idea is closely related to the green-beard effect \citep{Dawkins1976,Gardner2010}. It is also connected to tag-based models of cooperation, in which behavior
depends directly on a visible heritable marker rather than on explicit
pedigree information
\citep{Riolo2001,Axelrod2004,Traulsen2007}.
Related evolutionary models have shown that tag-based and
similarity-dependent interactions can generate cooperation and phenotypic
segregation
\citep{TraulsenSchuster2003,TraulsenClaussen2004}.
Many of these models assume that a helper cooperates fully with those
carrying the same tag and defects against those carrying a different tag.

Phenotype-space models make this idea more explicit by placing recognition traits in a metric space. \citet{Antal2009} studied cooperation based on phenotypic similarity in a one-dimensional integer phenotype space, where cooperators help only individuals with the same phenotype.  \citet{Kroumi2015} extended this approach to the multidimensional lattice \(\mathbb Z^n\) and proposed that increasing the phenotype-space dimension promotes cooperation. These studies assume that a cooperator is able to detect phenotypic dissimilarity between itself and other individuals in the population. Recently, \citet{Kroumi2026} considered graded recognition, in which cooperation decreases with phenotypic distance rather than being restricted to identical phenotypes, for phenotype spaces of the form \(\mathbb Z^n\).

The integer lattice \(\mathbb Z^n\) is a useful idealization, but it allows infinitely many phenotypic states and arbitrarily large phenotypic distances. Real recognition systems are usually more constrained. Chemical recognition in social insects, for example, depends on cuticular-hydrocarbon profiles that are perceived through finite sensory systems and used to distinguish nestmates from non-nestmates \citep{Yusuf2010,vanWilgenburg2012,Pask2017}. In plant self-incompatibility systems, self/nonself recognition is controlled by highly polymorphic but discrete genetic determinants \citep{Nasrallah2019,Hua2008}. MHC-mediated recognition likewise involves discrete genetic variation with consequences for immune recognition \citep{Migalska2019}. Visible signals provide further examples: some species exhibit a finite number of genetically determined color morphs \citep{Brock2020,WhiteKemp2016ColourPolymorphism}, while perceptual color spaces are bounded by the sensory capacities of the receiver \citep{Renoult2017}. These systems differ biologically, but they share the feature that the phenotypic cues relevant to recognition are finite, bounded, or effectively discrete.

A natural representation of such a bounded phenotype space is the discrete torus $\mathbb T_L^n=\{0,1,\ldots,L-1\}^n$. It contains finitely many phenotypic states, permits mutation by local movement between neighboring states, and avoids artificial boundary effects by making each coordinate periodic. This means that the torus retains the local geometry of the lattice while replacing its unbounded state space by a finite and homogeneous one. The parameters \(L\) and \(n\) have distinct interpretations: \(L\) determines the number of states available in each coordinate, whereas \(n\) determines the number of phenotypic coordinates used in recognition.

The one-dimensional case \(\mathbb T_L\) is the discrete cycle and is a natural model for periodic recognition traits, such as orientation, activity phase, or seasonal timing, for which the two ends of a linear scale should be identified \citep{Landler2018CircularData}. A cycle may also represent a coarse periodic component of a more complex recognition cue. Thus, \(\mathbb T_L\) provides a simple finite analogue of a one-dimensional phenotype space while preserving local mutational movement.

Another important special case is the \(n\)-dimensional hypercube
\(\mathbb T_2^n=\{0,1\}^n\), obtained when \(L=2\). In this case, the periodic distance reduces to Hamming distance, namely the number of coordinates at which two phenotypes differ. Binary tags and Hamming-type recognition have been used in several models of conditional cooperation \citep{HalesEdmonds2005,Kim2010,Traulsen2007}. Binary phenotype spaces also arise naturally when recognition depends on several discrete molecular or genetic components, as in microbial kin discrimination and allorecognition \citep{Wall2016KinRecognition,Lyons2016CombinatorialKin,Hirose2015Allorecognition}. More generally, hypercube graphs are widely used to represent binary genotype and phenotype spaces in which one-step mutations alter a single coordinate \citep{Zagorski2016BeyondHypercube}.

The finite geometry also changes the analytical treatment of the phenotype
process. Although the coalescent framework is similar to that used for the
unbounded space \(\mathbb Z^n\) \citep{Kroumi2026}, the phenotypic
difference process is now a random walk on a finite periodic state space.
Its transition kernel is determined by the discrete spectrum of the random
walk on the torus, leading to a finite Fourier representation. This differs
from the analysis on the unbounded lattice and introduces finite-state
recurrence and mixing that have no direct analogue in \(\mathbb Z^n\).

In this paper, we analyze a phenotypic-recognition model on the torus, where
individuals with similar phenotypes are more likely to provide cooperative
benefits to one another, whereas phenotypically distant individuals receive
smaller benefits. Under weak selection, we combine the Kingman coalescent
with the spectral representation of the periodic random walk to derive an
explicit threshold for the benefit-to-cost ratio above which selection favors
the abundance of cooperation. We prove monotonicity results for the effects
of discrimination, phenotype-space dimension, and strategy mutation. The
finite geometry also produces a qualitative difference from the unbounded
model: the cooperation threshold diverges both for very small and for very
large phenotype mutation rates and therefore attains at least one minimum at
an intermediate mutation rate. At high mutation rates, this increase results
from mixing on the finite phenotype space, which weakens the genealogical
information carried by phenotypic similarity.

\section{Methods \label{sec1}}


\subsection{\label{sec1-1}Model}
Consider a finite haploid population of size $N$, labeled by $1,\ldots,N$. Individual $k$ is characterized by a strategy $S(k)\in\{C,D\}$, where $C$ and $D$ stand for cooperation and defection, respectively. It also carries a phenotype
\[
\mathbf y(k)=(y_1(k),\ldots,y_n(k))\in \mathbb{T}_L^n,
\]
where $\mathbb{T}_L=\{0,1,\ldots,L-1\}$.
This means that each phenotype is given by an \(n\)-dimensional vector with \(L\) possible states in each coordinate. The population is well-mixed in the sense that each individual can interact with every other individual in the population. However, an individual's behavior depends on its phenotypic dissimilarity from its partner. To measure the similarity, we use the periodic distance $d_n$ defined as follows. Consider two individuals $k$ and $j$ with phenotypes $\mathbf{y}(k)$ and $\mathbf{y}(j)$, then 
\begin{equation}\label{eq:distance}
d_n(k,j)=\sum_{r=1}^n d_1\bigl(y_r(k),y_r(j)\bigr),
\end{equation}
where
\begin{equation}\label{eq:distance1}
d_1(x,y)=\min\bigl(|y-x|,\,L-|y-x|\bigr)
\end{equation}
is the shortest distance between sites \(x\) and \(y\) on the circle \(\mathbb{T}_L\).

For social interactions, we assume pairwise interactions according to the simplified Prisoner's Dilemma. When an individual cooperates, it pays a cost \(c>0\) to provide a benefit \(b>c\) to its partner. If it decides not to cooperate, it never pays this cost and never provides the benefit to its partner.

An individual \(k\) with strategy \(C\) conditions its action on the phenotypic distance from its partner \(j\). More precisely, it helps \(j\) with probability
\begin{equation}\label{eq:kernel}
\varphi(k,j)=e^{-\alpha d_n(k,j)}.
\end{equation}
Because \(\varphi(k,j)=1\) when the two individuals have the same phenotype, individual $k$ always helps any partner with the same phenotype.  When their phenotypes differ, the recognition system of \(k\) fails to detect the dissimilarity with probability \(\varphi(k,j)\), in which case \(k\) helps \(j\). With the complementary probability \(1-\varphi(k,j)\), the dissimilarity is detected and \(k\) withholds help. The parameter \(\alpha>0\) measures the strength of discrimination: cooperation decreases slowly with phenotypic distance when \(\alpha\) is small and rapidly when \(\alpha\) is large. Thus, increasing \(\alpha\) concentrates cooperative acts on phenotypically similar partners. An individual with strategy \(D\) defects independently of its phenotypic similarity to its partner. See Fig.~\ref{fig1} for an illustration.
\begin{figure*}
  \centering
  \includegraphics[width=\textwidth]{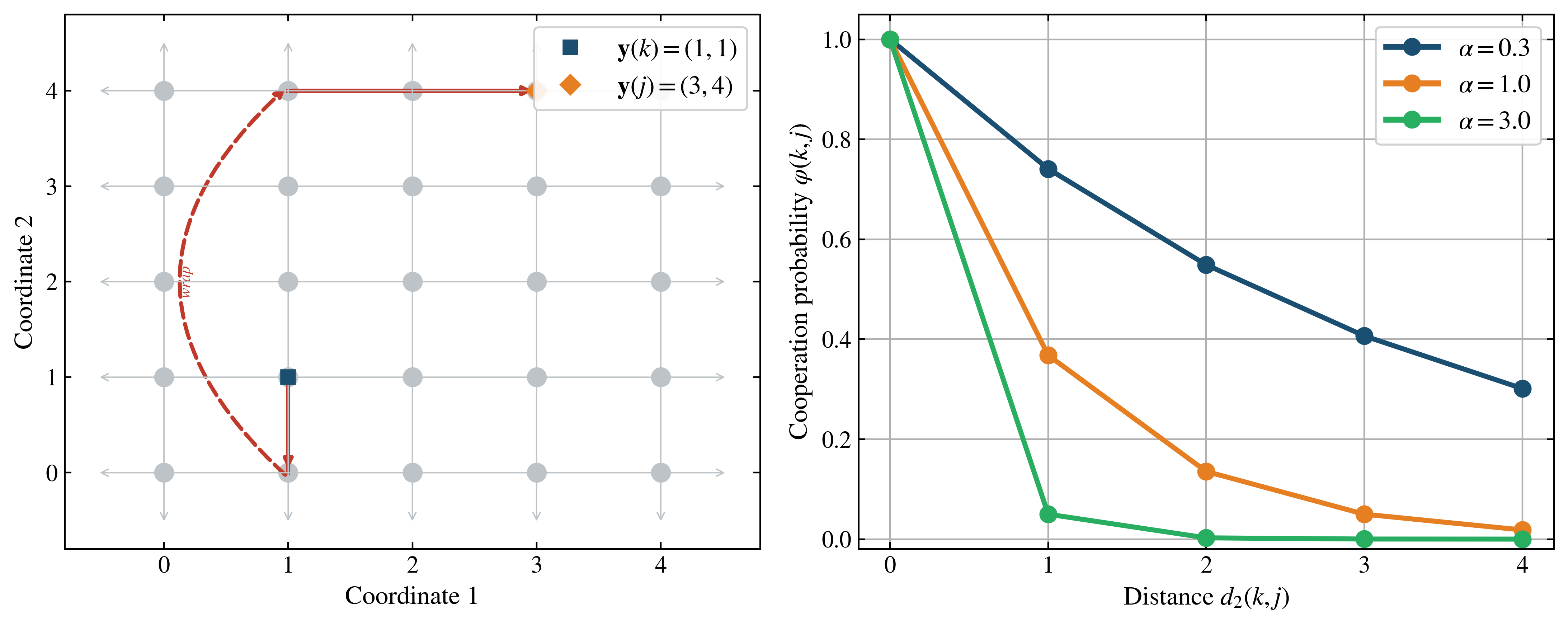}
  \caption{The periodic phenotype space $\mathbb{T}_5^2$ for $n=2$ and $L=5$, and its associated cooperation kernel. 
\textbf{Left:} each coordinate wraps around at $L=5$, forming a torus. As a result, the shortest route from phenotype $(1,1)$ to phenotype $(3,4)$ goes one step downward, crosses the periodic boundary, and then moves two steps to the right, for a total distance of $4$. 
\textbf{Right:} the cooperation probability $\varphi(k,j)$ shown as a function of the phenotypic distance for three values of the discrimination parameter $\alpha$. Since the largest possible distance on $\mathbb{T}_5^2$ is four, there are only five distinct cooperation levels. Larger values of $\alpha$ make the decline steeper, focusing cooperation more strongly on the closest phenotypic neighbors.}
  \label{fig1}
\end{figure*}

After interactions, the expected payoff to individual \(k\) can be written as
\begin{equation}\label{eq:payoff}
\omega_k=\frac{1}{N-1}\sum_{j\ne k}
\left[
b\,\varphi(j,k)\mathbf 1_{\{S(j)=C\}}
-c\,\varphi(k,j)\mathbf 1_{\{S(k)=C\}}
\right],
\end{equation}
where \[\mathbf 1_{\{S(j)=C\}}
=
\begin{cases}
1, & \text{if } S(j)=C,\\
0, & \text{otherwise}.
\end{cases}\]
Here, the first argument of \(\varphi\) is the actor and the second is the recipient. Individual \(k\) receives a benefit from each cooperating partner \(j\), discounted by the probability \(\varphi(j,k)\) that the recognition system of \(j\) does not detect their phenotypic dissimilarity. Individual \(k\) pays the cost \(c\) in an interaction with partner \(j\) when \(k\) is a cooperator and its own recognition system fails to detect their phenotypic dissimilarity. Since \(d_n(k,j)=d_n(j,k)\), the present kernel is symmetric, where
\[
\varphi(j,k)=\varphi(k,j).
\]
We nevertheless retain the actor--recipient ordering in Eq.~\eqref{eq:payoff} to keep the two roles explicit. Then, the expected payoff received by each individual is translated into a reproductive value given by
\begin{equation}\label{eq:fertility}
f_k=1+\delta \omega_k,
\end{equation}
where $\delta>0$ is the selection intensity.  This means that each individual has a baseline reproductive value given by $1$ and an additional value given by its advantage in interactions with others governed by $\delta$. We assume throughout that \(\delta\) is sufficiently small that \(f_k>0\) for every individual and every population configuration.

For reproduction, we use a discrete-time Moran model. At each time step, individuals interact, and each individual accumulates a payoff that is translated into a reproductive value. Then, one individual is chosen with probability proportional to its reproductive value to produce an offspring. This offspring inherits the parent's strategy with probability $1-u$, where $0<u<1$. With the complementary probability $u$, it chooses a strategy randomly selected among $\{C,D\}$. For the phenotype, the offspring inherits the parent's phenotype with probability $1-v$, where $0<v<1$. With the complementary probability $v$, a phenotypic mutation occurs in exactly one coordinate $r\in\{1,2,\ldots,n\}$, while all the other coordinates remain the same as those of the parent's phenotype. In that case, coordinate $r$ moves one step forward with probability $1/2$, or one step backward with probability $1/2$. More precisely, with probability $v/(2n)$ the phenotype of the offspring is $\mathbf y(k)+\mathbf e_r \pmod L$, and with probability $v/(2n)$ it is $\mathbf y(k)-\mathbf e_r \pmod L$,
where $\mathbf e_r$ is the $r$-th unit vector. We assume that mutation events act independently on the two traits, strategy and phenotype. Finally, the offspring replaces an individual chosen at random from the population and takes its label in the next time step.

Now consider the phenotype--strategy configuration of the entire population at
time \(t\), denoted by
\begin{equation}
\mathbf D_t=\bigl((\mathbf y_t(1),S_t(1)),\ldots,(\mathbf y_t(N),S_t(N))\bigr).
\end{equation}
Here, \(\mathbf y_t(i)\) and \(S_t(i)\) denote, respectively, the phenotype and strategy of individual \(i\) at time \(t\), for \(i=1,\ldots,N\). The process \((\mathbf D_t)_{t\geq0}\) takes values in the finite state space $\mathcal E=\bigl(\mathbb T_L^n\times\{C,D\}\bigr)^N$, where any configuration can be reached from any other through a finite sequence of reproduction events, nearest-neighbor phenotype mutations, and strategy mutations. Hence, \((\mathbf D_t)_{t\geq0}\) is an irreducible Markov chain on the finite state space \(\mathcal E\). It therefore admits a unique stationary distribution, denoted by $\bigl(\pi_\delta(\mathbf d)\bigr)_{\mathbf d\in\mathcal E}$. We write \(\E_\delta[\cdot]\) for expectation with respect to this stationary distribution.

Under neutrality, that is, when \(\delta=0\), the transition rule is equivariant under translations of the phenotype torus. Hence, each one-lineage phenotype marginal is uniform on \(\mathbb{T}_L^n\). The joint neutral distribution, however, is not product-uniform, because neutral reproduction copies phenotypes and creates genealogical correlations among individuals. We denote expectation with respect to the neutral stationary distribution by \(\E_0[\cdot]\), and we compute the required joint and relative distributions using a coalescent approach. Since the state space is finite and the transition probabilities depend smoothly on \(\delta\), the stationary distribution also depends smoothly on \(\delta\) for \(\delta\) near \(0\). Hence, we have 
\begin{equation}\label{eq:delta_neutral}
\E_\delta[\cdot]=\E_0[\cdot]+o(1).
\end{equation}

\subsection{Weak-selection analysis}
In this section, we compute the average abundance of cooperation at stationarity of the process \((\mathbf D_t)_{t\geq0}\). We then derive the condition on the ratio $b/c$ for weak selection to favor the average abundance of cooperation.

Suppose that the population is in configuration \(\mathbf d\in \mathcal{E}\) at a given time
step at stationarity, and denote by
\[
X(\mathbf d)=\frac1N\sum_{k=1}^N \mathbf 1_{\{S(k)=C\}}
\]
the frequency of cooperators in that configuration. Any individual
\(k\in\{1,2,\ldots,N\}\) is selected to produce an offspring with probability
\begin{equation}\label{eq:probP}
p_k(\mathbf d)
=
\frac{1+\delta\omega_k(\mathbf d)}
{\sum_{\ell=1}^N \bigl(1+\delta\omega_\ell(\mathbf d)\bigr)}.
\end{equation}

When strategy mutation occurs, with probability \(u\), the offspring is of
type \(C\) or type \(D\), each with probability \(1/2\). This offspring
replaces an individual chosen uniformly at random from the population. Under
this scenario, the frequency of type \(C\) increases (respectively decreases)
by \(1/N\) if the offspring is of type \(C\) (respectively type \(D\)) and
the individual chosen for replacement is of type \(D\) (respectively type
\(C\)). Hence, the conditional expected change in the frequency of type \(C\)
in one time step, denoted by \(\Delta X_{\mathrm{mut}}\), is
\begin{equation}
\begin{split}
\E_\delta\left[\Delta X_{\mathrm{mut}}\mid \mathbf d\right]&=\frac{1}{N}\left(\frac{1}{2}(1-X(\mathbf d))-\frac{1}{2}X(\mathbf d)\right)\\
&=\frac{1}{N}\left(\frac{1}{2}-X(\mathbf d)\right).
\end{split}
\end{equation}

In the absence of strategy mutation, which occurs with probability \(1-u\),
the frequency of cooperators increases (respectively decreases) by \(1/N\)
if a cooperator (respectively a defector) is selected to reproduce and its
offspring replaces a defector (respectively a cooperator). In this case, the
conditional expected change in the frequency of cooperation in one time
step, denoted by \(\Delta X_{\mathrm{sel}}\), is
\begin{align}
\E_\delta\left[\Delta X_{\mathrm{sel}}\mid \mathbf d\right]&=\frac{1}{N}\left[\sum_{k=1}^Np_k(\mathbf d)\mathbf 1_{\{S(k)=C\}}(1-X(\mathbf d))\right.\nonumber\\
&\quad\left.-\sum_{k=1}^Np_k(\mathbf d)\mathbf 1_{\{S(k)=D\}}X(\mathbf d)\right]\label{eq:selectioneffect}\\
&=\frac{1}{N}\left[\sum_{k=1}^Np_k(\mathbf d)\mathbf 1_{\{S(k)=C\}}-X(\mathbf d)\right].\nonumber
\end{align}
Here, we have used $\sum_{k=1}^N p_k(\mathbf d)=1$ and $\mathbf 1_{\{S(k)=D\}}=1-\mathbf 1_{\{S(k)=C\}}$.

Finally, conditioning on the presence/absence of strategy mutation, the expected change in the frequency of cooperation in one time step can be written as
\begin{align}
\E_\delta\left[\Delta X\mid \mathbf d\right]&=(1-u)\E_\delta\left[\Delta X_{\mathrm{sel}}\mid \mathbf d\right]+u\E_\delta\left[\Delta X_{\mathrm{mut}}\mid \mathbf d\right]\nonumber\\
&=(1-u)\,\E_\delta[\Delta X_{\mathrm{sel}}\mid \mathbf d]+\frac{u}{N}\left(\frac{1}{2}-X(\mathbf d)\right).\label{eq:decomp}
\end{align}
Taking the expectation according to the distribution $(\pi_{\delta}(\mathbf d))_{\mathbf{d}\in\mathcal{E}}$, we get
\begin{equation}\label{eq100}
 \E_\delta\left[\Delta X\right]=(1-u)\,\E_\delta[\Delta X_{\mathrm{sel}}]+\frac{u}{N}\left(\frac{1}{2}- \E_\delta\left[X\right]\right). 
\end{equation}
At stationarity, note that the average abundance of type \(C\) is constant, which is equivalent to
\[
\E_\delta[\Delta X]=0.
\]
Then, Eq.~\eqref{eq100} yields
\begin{equation}\label{eq:abundance}
\E_\delta[X]=\frac12+\frac{N(1-u)}{u}\E_\delta[\Delta X_{\mathrm{sel}}].
\end{equation}

Under the neutral regime $\delta=0$, Eq.~\eqref{eq:probP} gives $p_k(\mathbf d)=1/N$ for any individual $k$ and any state $\mathbf{d}$. Then, from \eqref{eq:selectioneffect}, we deduce that the expected change in the frequency of cooperation in one time step in the absence of strategy mutation vanishes $\E_0[\Delta X_{\mathrm{sel}}]=0$ leading to
\[
\E_0[X]=\frac12.
\]
Thus, weak selection favors the abundance of cooperation when the average abundance of $C$, $\E_\delta[X]$, exceeds its neutral value $\E_0[X]=1/2$. Using Eq.~\eqref{eq:abundance}, this is equivalent to
\[
\E_\delta[\Delta X_{\mathrm{sel}}]>0.
\]

To evaluate this quantity, we expand $p_k(\mathbf d)$ given in Eq.~\eqref{eq:probP} to first order in $\delta$ as
\[
p_k(\mathbf d)=\frac1N+\frac{\delta}{N}(\omega_k(\mathbf d)-\bar\omega(\mathbf d))+O(\delta^2),
\]
where $\bar\omega(\mathbf{d})=\frac1N\sum_{\ell=1}^N \omega_\ell(\mathbf{d})$ is the average payoff in the population at state $\mathbf{d}$.
Substituting this expression in Eq.~\eqref{eq:selectioneffect} and using Eq.~\eqref{eq:payoff} and \(\varphi(j,k)=\varphi(k,j)\), we obtain 
\begin{align}
&\E_\delta\left[\Delta X_{\mathrm{sel}}\mid \mathbf d\right]\nonumber\\
&=\frac{\delta}{N^2}\left(\sum_{k}\mathbf 1_{\{S(k)=C\}}\omega_k(\mathbf d)-
X(\mathbf d)\sum_{k}\omega_k(\mathbf d)
\right)+\mathcal{O}(\delta^2)\nonumber\\
&=\frac{\delta}{N^2(N-1)}\Bigg( b\sum_{k\ne j}\varphi(k,j)\mathbf 1_{\{S(k)=S(j)=C\}}\nonumber\\
 & \qquad-\frac{b-c}{N}\sum_{k\ne j}\sum_{\ell}\varphi(k,j)\mathbf 1_{\{S(k)=S(\ell)=C\}} \nonumber\\
 &\qquad-c\sum_{k\ne j}\varphi(k,j)\mathbf 1_{\{S(k)=C\}} \Bigg)+\mathcal{O}(\delta^2)\label{eq:selchange}
\end{align}
Taking the expectation and using Eq.~\eqref{eq:delta_neutral}, we obtain
\begin{equation}\label{eq:sel_change}
\begin{split}
 & \E_\delta\left[\Delta X_{\mathrm{sel}}\right]\\
  =&\frac{\delta}{N^2(N-1)}\Bigg( b\sum_{k\ne j}\E_0\left[\varphi(k,j)\mathbf 1_{\{S(k)=S(j)=C\}}\right]\\
 & \qquad-\frac{b-c}{N}\sum_{k\ne j}\sum_{\ell}\E_0\left[\varphi(k,j)\mathbf 1_{\{S(k)=S(\ell)=C\}}\right]\\
 &\qquad -c\sum_{k\ne j}\E_0\left[\varphi(k,j)\mathbf 1_{\{S(k)=C\}}\right] 
  \Bigg)+o(\delta)
\end{split}
\end{equation}

 Under neutrality, note that strategies are equally likely to be $C$ or $D$. Then, we can express the needed expectations in Eq.~\eqref{eq:sel_change} in terms of the three basic identity measures defined as 
\begin{subequations}\label{identities}
\begin{align}
\ZnL(\alpha) &:= \E_0[\varphi(k,j)],\label{eq:Zdef}\\
\GnL(\alpha) &:= \E_0\bigl[\varphi(k,j)\mathbf 1_{\{S(k)=S(j)\}}\bigr],\label{eq:Gdef}\\
\HnL(\alpha) &:= \E_0\bigl[\varphi(k,j)\mathbf 1_{\{S(j)=S(\ell)\}}\bigr],\label{eq:Hdef}
\end{align}
\end{subequations}
where $k$, $j$ and $\ell$ designate three distinct individuals chosen at random at the same time step in the neutral population at stationarity.
The identity $Z_n^{(L)}(\alpha)$ measures the joint phenotypic kernel between individuals $k$ and $j$, while $G_n^{(L)}(\alpha)$ captures both their joint phenotypic kernel and strategic correlation. Lastly, $H_n^{(L)}(\alpha)$ measures the joint phenotypic kernel between individuals $k$ and $j$ and the strategic correlation between individuals $j$ and $\ell$.

Using neutral \(C\leftrightarrow D\) symmetry and exchangeability, we have
\begin{align*}
\sum_{k\ne j}
\E_0\left[
\varphi(k,j)\mathbf 1_{\{S(k)=C\}}
\right]
&=
\frac{1}{2}
\sum_{k\ne j}
\E_0\left[\varphi(k,j)\right]
\\
&=
\frac{N(N-1)}{2}\ZnL(\alpha).
\end{align*}
Similarly,
\begin{align*}
&\sum_{k\ne j}
\E_0\left[
\varphi(k,j)\mathbf 1_{\{S(k)=S(j)=C\}}
\right]\\
&=
\frac{1}{2}
\sum_{k\ne j}
\E_0\left[
\varphi(k,j)\mathbf 1_{\{S(k)=S(j)\}}
\right]
\\
&=
\frac{N(N-1)}{2}\GnL(\alpha).
\end{align*}
For the third sum, distinguishing the cases
\(\ell=k\), \(\ell=j\), and \(\ell\notin\{k,j\}\), we obtain
\begin{align*}
&\sum_{k\ne j}\sum_{\ell}
\E_0\left[
\varphi(k,j)
\mathbf 1_{\{S(k)=S(\ell)=C\}}
\right]
\\
&\qquad=
\frac{N(N-1)}{2}
\Bigl[
\ZnL(\alpha)
+\GnL(\alpha)
+(N-2)\HnL(\alpha)
\Bigr].
\end{align*}

Inserting these expressions in Eq.~\eqref{eq:sel_change}, the expected change in the frequency of cooperation in the absence of strategy mutation can be written as
\begin{equation}\label{eq:finalsel}
\begin{split}
&\E_\delta[\Delta X_{\mathrm{sel}}]\\
=&\frac{\delta}{2N^2}\Bigl\{b\bigl[(N-1)\GnL-\ZnL-(N-2)\HnL\bigr]\\
&-c\bigl[(N-1)\ZnL-\GnL-(N-2)\HnL\bigr]\Bigr\}+o(\delta).
\end{split}
\end{equation}
Then, we deduce that weak selection favors the average abundance of cooperation to first order if and only if
\begin{equation}\label{eq:condition}
\begin{split}
&b\bigl[(N-1)\GnL(\alpha)-\ZnL(\alpha)-(N-2)\HnL(\alpha)\bigr]\\
&>c\bigl[(N-1)\ZnL(\alpha)-\GnL(\alpha)-(N-2)\HnL(\alpha)\bigr],
\end{split}
\end{equation}
which is equivalent to
\[
\frac{b}{c}>\beta_{n}^{(L)}(\alpha),
\]
where
\begin{equation}\label{eq:betafinite}
\beta_{n}^{(L)}(\alpha)=\frac{(N-1)\ZnL(\alpha)-\GnL(\alpha)-(N-2)\HnL(\alpha)}{(N-1)\GnL(\alpha)-\ZnL(\alpha)-(N-2)\HnL(\alpha)}.
\end{equation}

This formula is exact, but it is not yet sufficiently explicit in terms of the model parameters to facilitate further analysis. To obtain a simpler expression, consider the large-population limit \(N\to\infty\), with \(u\to0\) and \(v\to0\) such that the scaled mutation rates $\mu=Nu/2$ and $\nu=Nv/2$ are kept fixed. In this limit, the threshold becomes
\begin{equation}\label{eq:betalarge}
\bc(\alpha)=\frac{\ZnL(\alpha)-\HnL(\alpha)}{\GnL(\alpha)-\HnL(\alpha)}.
\end{equation}

\section{Threshold formula}\label{sec:threshold}

We now derive an explicit expression for the large-population threshold
\(\beta_n^{(L)}(\alpha)\) appearing in Eq.~\eqref{eq:betalarge}.
The detailed calculations are given in
Appendix~\ref{app:threshold_derivation}.

Measuring time in units of \(N/2\) generations, the ancestral lineages
converge, as \(N\to\infty\), to the Kingman coalescent
\citep{Kingman1982Genealogy}. Along the branches of the coalescent tree,
strategy and phenotype mutations converge to two independent Poisson
processes with rates \(\mu\) and \(\nu\), respectively. For a related
derivation, see \citet{Kroumi2015}.

Consider two individuals \(k\) and \(j\) sampled uniformly at random from
the population at the same time in a neutral population at stationarity, and trace their
ancestral lineages backward in time. Let
\(T_{2}=T(k,j)\) denote their coalescence time to their most recent
common ancestor. Then
\[
f_{T_{2}}(\tau)=e^{-\tau},
\qquad \tau>0.
\]
For three sampled individuals \(k,j,\ell\), let \(T_{3}\) denote the
time until the first coalescence event and let \(T_{2}\) denote the
additional time until the remaining two lineages coalesce. These two
random variables are independent, with joint density
\[
f_{T_{2},T_{3}}(\tau_2,\tau_3)
=
3e^{-(3\tau_3+\tau_2)},
\qquad
\tau_2,\tau_3>0.
\]

Note that strategy mutations occur independently at rate \(\mu\) along each ancestral lineage. Conditional on a coalescence time \(T_{2}=\tau\), the total length of the two ancestral branches is \(2\tau\). Therefore, two individuals \(k\) and \(j\), sampled from the population at the same time, have the same strategy with probability
\begin{equation}\label{eq:strategyprob}
\PP_0\bigl(S(k)=S(j)\mid T_{2}=\tau\bigr)
=\frac{1+e^{-2\mu\tau}}{2}.
\end{equation}
Indeed, no strategy mutation occurs on either ancestral branch with probability \(e^{-2\mu\tau}\), in which case both individuals inherit the strategy of their most recent common ancestor. With probability \(1-e^{-2\mu\tau}\), at least one mutation occurs. Conditional on this event, consider the most recent mutation between the common ancestor and the time in which the individuals are sampled. Since this mutation redraws the strategy uniformly from \(\{C,D\}\), the resulting strategy agrees with that carried by the other lineage with probability \(1/2\). Hence
\[
\begin{aligned}
\PP_0\bigl(S(k)=S(j)\mid T_{2}=\tau\bigr)&=e^{-2\mu\tau}+\frac{1}{2}\bigl(1-e^{-2\mu\tau}\bigr)\\
&=\frac{1+e^{-2\mu\tau}}{2}.
\end{aligned}
\]

We next consider the phenotype process. For each coordinate
\(r\in\{1,\ldots,n\}\), the phenotypic difference between the two
ancestral lineages evolves as a continuous-time symmetric nearest-neighbor
random walk on \(\mathbb T_L\), with total jump rate \(2\nu/n\).
Its conditional transition probability is
\begin{equation}\label{eq:skellam}
\begin{split}
&\PP_0\bigl(y_r(k)-y_r(j)=m\pmod L\mid T_{2}=\tau\bigr) \\
=&\frac1L\sum_{p=0}^{L-1}e^{-\sigma(1-c_p)}e^{2\pi\mathrm{i}pm/L},
\end{split}
\end{equation}
for any $m\in\mathbb T_L$, where $\sigma=2\nu\tau/n$ and $c_p=\cos\!\left(2\pi p/L\right)$. Here, \(\mathrm{i}\) denotes the imaginary unit.
The derivation of Eq.~\eqref{eq:skellam} from the spectral decomposition
of the periodic random walk is given in
Appendix~\ref{app_1:threshold_derivation}.

Define the one-coordinate phenotypic function
\begin{equation}\label{eq:Phidef_main}
\Phi_\alpha^{(L)}(\sigma):=\E_0\left[e^{-\alpha d_1(0,y_r(k)-y_r(j))}\Big| T_{2}=\tau\right].
\end{equation}
Using \eqref{eq:skellam}, we get
\begin{equation}\label{eq:Phidef}
\begin{split}
&\Phi_\alpha^{(L)}(\sigma)\\
=&  \sum_{m=0}^{L-1} e^{-\alpha d_1(0,m)}\PP_0\Big(y_r(k)-y_r(j)=m\Big| T_{2}=\tau\Big)\\
=&\sum_{m=0}^{L-1} e^{-\alpha d_1(0,m)}\frac1L\sum_{p=0}^{L-1} e^{-\sigma(1-c_p)}e^{2\pi \mathrm{i} p m/L}\\
=&\frac1L\sum_{p=0}^{L-1}\hatphi(p)e^{-\sigma(1-c_p)}.
\end{split}
\end{equation}
where
\begin{equation}\label{eq:DFT_main}
\hatphi(p)
=
\sum_{m=0}^{L-1}
e^{-\alpha d_1(0,m)}
e^{2\pi\mathrm{i}pm/L}.
\end{equation}
The coefficients \(\hatphi(p)\) are strictly positive for \(p=0,\ldots,L-1\).
The explicit formulas for \(\hatphi(p)\), together with some properties of $\Phi_\alpha^{(L)}$, are given in
Appendix~\ref{app_2:threshold_derivation}.

Since the \(n\) coordinate processes are independent conditional on the
coalescence time, the full phenotypic kernel satisfies
\begin{equation}\label{eqqq}
\begin{split}
&\E_0\left[\varphi(k,j)\mid T_{2}=\tau\right]\\
=\;&\E_0\!\left[e^{-\alpha d_n(k,j)}\,\Big|\, T_{2}=\tau\right]\\
=\;&\prod_{r=1}^n
   \E_0\!\left[e^{-\alpha d_1(0,\,y_r(k)-y_r(j))}\,\Big|\,
   T_{2}=\tau\right]\\
=\;&\left[
\Phi_\alpha^{(L)}
\left(\frac{2\nu\tau}{n}\right)
\right]^n.
\end{split}
\end{equation}
 Conditioning on the corresponding coalescence times gives
\begin{subequations}\label{eq:ZGexplicit_main}
\begin{align}
Z_n^{(L)}(\alpha)=&
\frac{n}{2\nu}
\mathcal L_n^{(L)}
\left(
\alpha,\frac{1}{2\nu}
\right),
\label{eq:Zexplicit}
\\
G_n^{(L)}(\alpha)
=&
\frac{n}{4\nu}
\left[
\mathcal L_n^{(L)}
\left(
\alpha,\frac{1}{2\nu}
\right)
+
\mathcal L_n^{(L)}
\left(
\alpha,\frac{1+2\mu}{2\nu}
\right)
\right],
\label{eq:Gexplicit}\\
H_n^{(L)}(\alpha)
=&
\frac{n}{8\nu}
\Bigg[
\frac{3+2\mu}{1+\mu}
\mathcal L_n^{(L)}
\left(
\alpha,\frac{1}{2\nu}
\right)
+
\mathcal L_n^{(L)}
\left(
\alpha,\frac{1+2\mu}{2\nu}
\right)
\nonumber\\
&
-
\frac{\mu(3+2\mu)}
{(1+\mu)(1+2\mu)}
\mathcal L_n^{(L)}
\left(
\alpha,\frac{3+2\mu}{2\nu}
\right)
\Bigg],\label{eq:Hexplicit}
\end{align}
\end{subequations}
where $\mathcal L_n^{(L)}$ is a generalized Laplace transform given by
\begin{equation}\label{eq:Ldef}
\mathcal L_n^{(L)}(\alpha,x)
:=
\int_0^\infty
\left[\Phi_\alpha^{(L)}(\sigma)\right]^n
e^{-nx\sigma}\,d\sigma.
\end{equation}
The details of these coalescent calculations are given in
Appendix~\ref{app_identity_measures}.

Substituting Eqs.~\eqref{eq:ZGexplicit_main} into
Eq.~\eqref{eq:betalarge} yields the explicit threshold
\begin{widetext}
\begin{equation}\label{eq:betaexplicit}
\beta_n^{(L)}(\alpha)
=
\frac{
(1+2\mu)^2
\mathcal L_n^{(L)}
\left(\alpha,\frac1{2\nu}\right)
+
\mu(3+2\mu)
\mathcal L_n^{(L)}
\left(\alpha,\frac{3+2\mu}{2\nu}\right)
-
(1+\mu)(1+2\mu)
\mathcal L_n^{(L)}
\left(\alpha,\frac{1+2\mu}{2\nu}\right)
}{
(1+\mu)(1+2\mu)
\mathcal L_n^{(L)}
\left(\alpha,\frac{1+2\mu}{2\nu}\right)
+
\mu(3+2\mu)
\mathcal L_n^{(L)}
\left(\alpha,\frac{3+2\mu}{2\nu}\right)
-
(1+2\mu)
\mathcal L_n^{(L)}
\left(\alpha,\frac1{2\nu}\right)
}.
\end{equation}
\end{widetext}

The numerator minus the denominator in Eq.~\eqref{eq:betaexplicit} is
\[
2(1+\mu)(1+2\mu)
\left[
\mathcal L_n^{(L)}
\left(\alpha,\frac1{2\nu}\right)
-
\mathcal L_n^{(L)}
\left(\alpha,\frac{1+2\mu}{2\nu}\right)
\right],
\]
which is strictly positive because
\(x\mapsto\mathcal L_n^{(L)}(\alpha,x)\) is strictly decreasing. See Appendix~\ref{ln_is_decreasing} for the proof.
Consequently,
\[
\beta_n^{(L)}(\alpha)>1.
\]
Thus, the benefit-to-cost ratio must always exceed \(1\) before weak
selection can favor the abundance of cooperation.

\section{Main results}
In this section, we study how the parameters of the evolutionary process affect the large-population threshold \(\beta_n^{(L)}(\alpha)\), with the remaining parameters held fixed unless stated otherwise.

\subsection{Effect of phenotypic discrimination}
Here, we investigate how the strength of phenotypic discrimination, controlled by the parameter $\alpha>0$, affects the threshold \(\beta_n^{(L)}(\alpha)\). The following result summarizes the effect of \(\alpha\) on the critical threshold. Its proof is given in Appendix~\ref{AppendixB}.

\begin{result}\label{thm:torus-alpha-monotonicity}
The map 
\(\alpha\mapsto\beta_n^{(L)}(\alpha)\) is strictly decreasing on \((0,\infty)\). Thus, sharper phenotypic discrimination lowers the threshold required for selection to favor the abundance of cooperation.
\end{result}

This result shows that cooperation is promoted most effectively when benefits are directed toward phenotypically close individuals ($\alpha\to\infty$). Although weak discrimination allows an individual to help a larger set of phenotypes, those additional benefits are spread across individuals that are less phenotypically associated with the actor. This weakens the assortment generated by phenotypic similarity. Increasing \(\alpha\) concentrates cooperative benefits among individuals with similar phenotypes. In this case, the recipient set becomes smaller, but the cooperative act becomes more strongly aligned with phenotypic assortment. 

Now, we describe two interesting limiting regimes.
\begin{itemize}
\item \textbf{Weak-discrimination limit \((\alpha\to0^+)\).}
When \(\alpha\) approaches zero, the exponential discrimination kernel becomes nearly constant across the torus:
\[
\varphi(k,j)=e^{-\alpha d_n(k,j)} \sim 1 .
\]
Thus, individuals distribute cooperative benefits almost independently of phenotypic similarity with their partners. This removes any phenotype-dependent bias in helping, and the assortment mechanism generated by recognition is lost. In this case, we have
\[
\lim_{\alpha\to0^+}\beta_n^{(L)}(\alpha)=\infty,
\]
which corresponds to phenotype-independent helping, where the recognition mechanism generates no assortment. For the proof, see Corollary~\ref{cor:alpha}.
Biologically, this means that indiscriminate helping does not favor the abundance of cooperation, because cooperative benefits are not preferentially directed toward phenotypically similar individuals. By contrast, individuals with similar phenotypes are more likely to share a recent common ancestor and, consequently, are more likely to carry the same strategy. Phenotypic similarity therefore provides an informative signal of strategic similarity, allowing cooperative benefits to be directed preferentially toward other cooperators.

\item \textbf{Sharp-discrimination limit \((\alpha\to\infty)\).}
At the opposite extreme, if $\alpha$ is very large, then the discrimination kernel converges to that of the exact-matching model, where
\[
\varphi(k,j)=e^{-\alpha d_n(k,j)}\sim\mathbf 1_{\{d_n(k,j)=0\}} .
\]
Thus, cooperation is directed only toward individuals with the same phenotype. By Result~\ref{thm:torus-alpha-monotonicity}, this threshold is the limiting lowest value of \(\beta_n^{(L)}(\alpha)\) for $\alpha\in(0,\infty)$.
\end{itemize}

Together, these limiting regimes clarify the role of phenotypic discrimination on the torus. Weak discrimination spreads benefits broadly but erases the effective association between helping and phenotypic similarity. Strong discrimination narrows the set of beneficiaries but strengthens assortment. In addition, the monotone decrease of \(\beta_n^{(L)}(\alpha)\) shows that the gain from stronger assortment outweighs the loss from helping fewer phenotypes.

These conclusions are illustrated in Fig.~\ref{fig2}.

\begin{figure}
  \centering
  \includegraphics[height=7cm, width=\columnwidth]{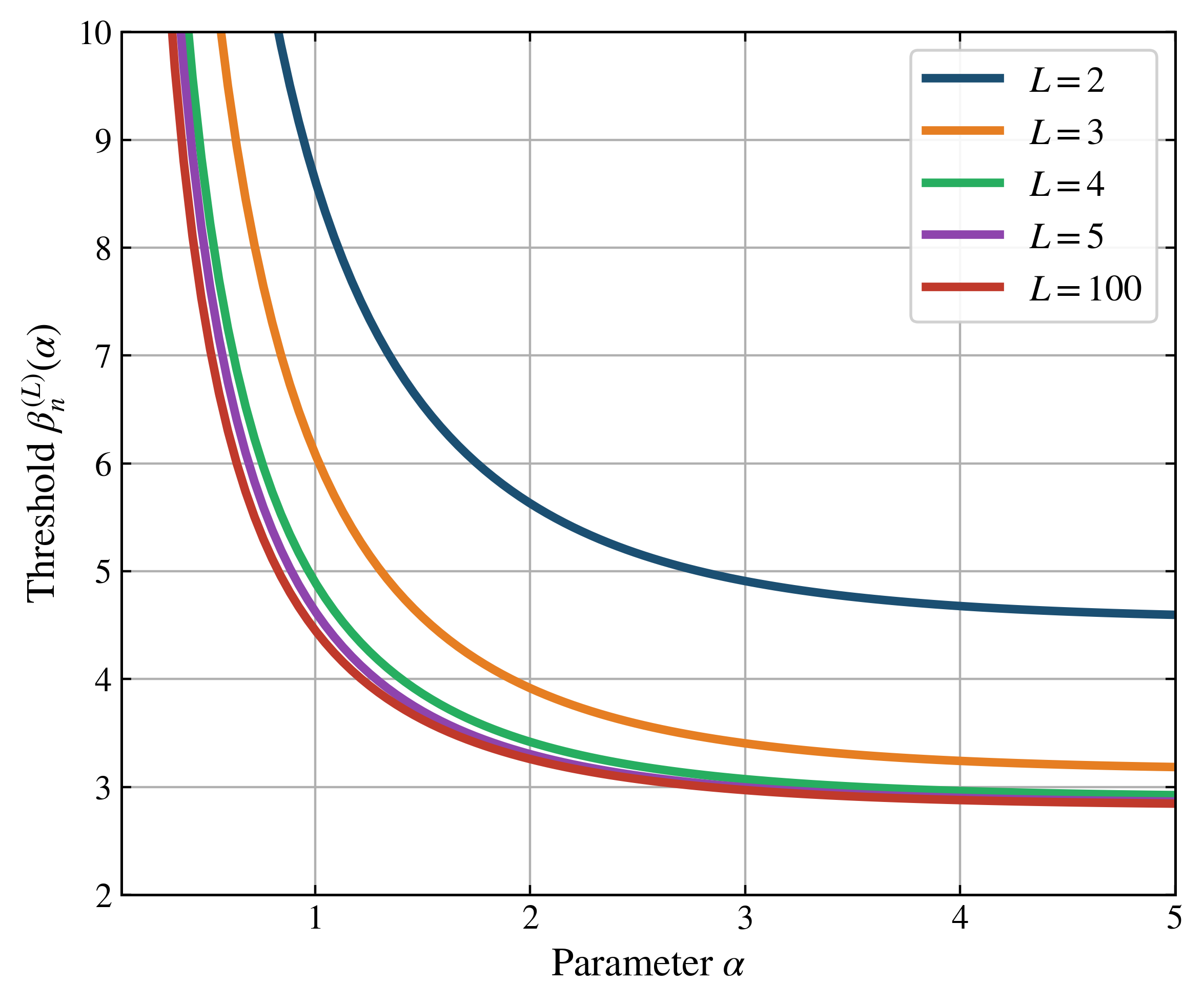}
  \caption{
  Threshold \(\beta_n^{(L)}(\alpha)\) as a function of the discrimination parameter \(\alpha\), with \(n=2\), \(\mu=0.5\), and \(\nu=1\), for \(L=2,3,4,5,100\). For each value of \(L\), the threshold decreases as \(\alpha\) increases, showing that stronger phenotypic discrimination makes cooperation easier to favor. For the parameter values displayed, the curves are also ordered downward as $L$ increases, suggesting a diminishing finite-state effect for larger $L$.
}
  \label{fig2}
\end{figure}

\subsection{\label{sec:torus-dimension}Effect of phenotype-space dimension}
Next, we study the effect of the dimension of the phenotype space $n\geq 1$ on the threshold \(\beta_n^{(L)}(\alpha)\). Note that increasing \(n\) adds independent phenotypic directions along which mutations can occur. This changes the geometry of phenotypic relatedness: two lineages may remain close in some coordinates while separating in others, and the total phenotypic distance reflects the accumulation of these coordinate-wise differences.

We have the following result, where the proof is given in Appendix~\ref{AppendixC}.

\begin{result}\label{thm:torus-n-monotonicity-section}
We have
\[
\beta_{n+1}^{(L)}(\alpha)<\beta_n^{(L)}(\alpha)
\qquad\text{for every integer } n\ge 1.
\]
\end{result}

This means that among phenotype spaces of this type, the one-dimensional torus is the least favorable for
the abundance of cooperation, while adding further phenotypic coordinates strictly lowers the threshold. The reason
is not that higher dimension directly increases the number of cooperative acts. Rather, it changes
how phenotypic similarity is distributed. In one dimension, mutational differences accumulate along
a single circular trait axis, so distinct lineages have relatively more overlap in their phenotypic
neighborhoods. In higher dimensions, the same mutational process is spread across several independent
coordinates. Lineages that acquire mutations in different coordinates separate more efficiently,
and the set of individuals that remain phenotypically close becomes more selectively filtered.

This geometric filtering strengthens assortment. Individuals that remain close in \(\mathbb{T}_L^n\) for large values of $n$
are more likely to share recent phenotypic history, whereas phenotypically distant individuals
receive rapidly diminishing cooperative benefits. Consequently, increasing \(n\) makes the phenotype
space a more effective mediator of assortment, which lowers the benefit-to-cost threshold.

In the limit $n\to\infty$, we have (see Theorem~\ref{cor:n_infty} for the proof) 
\begin{equation}\label{eq:torus-beta-infty-alpha-section}
\begin{split}
\beta_\infty^{(L)}(\alpha)&:=\lim_{n\to\infty}\beta_n^{(L)}(\alpha)\\
&=\frac{4\gamma_\alpha^2\nu^2+(8\mu+6)\gamma_\alpha\nu+4\mu^2+8\mu+3}{4\gamma_\alpha\nu\bigl(\gamma_\alpha\nu+\mu+1\bigr)}.
\end{split}
\end{equation}
Here \(\gamma_\alpha=1-e^{-\alpha}\). 
Note that this limiting threshold still depends on the discrimination parameter \(\alpha\). This means that increasing the dimension does not remove the role of phenotypic discrimination. Instead, dimension and discrimination act through complementary mechanisms. Larger \(n\) makes phenotypic neighborhoods geometrically thinner, while larger \(\alpha\) makes cooperative benefits more sharply concentrated inside those neighborhoods.

When discrimination vanishes,  we have \(\lim_{\alpha\to0^+}\gamma_\alpha=0\), and then Eq.~\eqref{eq:torus-beta-infty-alpha-section} yields
\[
\lim_{\alpha\to0^+}\beta_\infty^{(L)}(\alpha)=\infty.
\]
Even in very high dimension, cooperation cannot be favored if helping is essentially independent
of phenotypic distance. Conversely, in the sharp-discrimination limit, we have \(\lim_{\alpha\to\infty}\gamma_\alpha=1\), from which we deduce that 
\begin{equation}\label{eq:torus-beta-infty-infty-section}
\begin{split}
\beta_\infty^{(L)}(\infty)&:=\lim_{\alpha\to\infty} \beta_\infty^{(L)}(\alpha) \\
&=\frac{4\nu^2+(8\mu+6)\nu+4\mu^2+8\mu+3}{4\nu(\nu+\mu+1)}.
\end{split}
\end{equation}
This is the lowest threshold obtained by simultaneously taking the high-dimensional limit and the
exact-matching discrimination limit.

These conclusions are illustrated in Fig.~\ref{fig3}.

\begin{figure}
  \centering
  \includegraphics[height=7cm, width=\columnwidth]{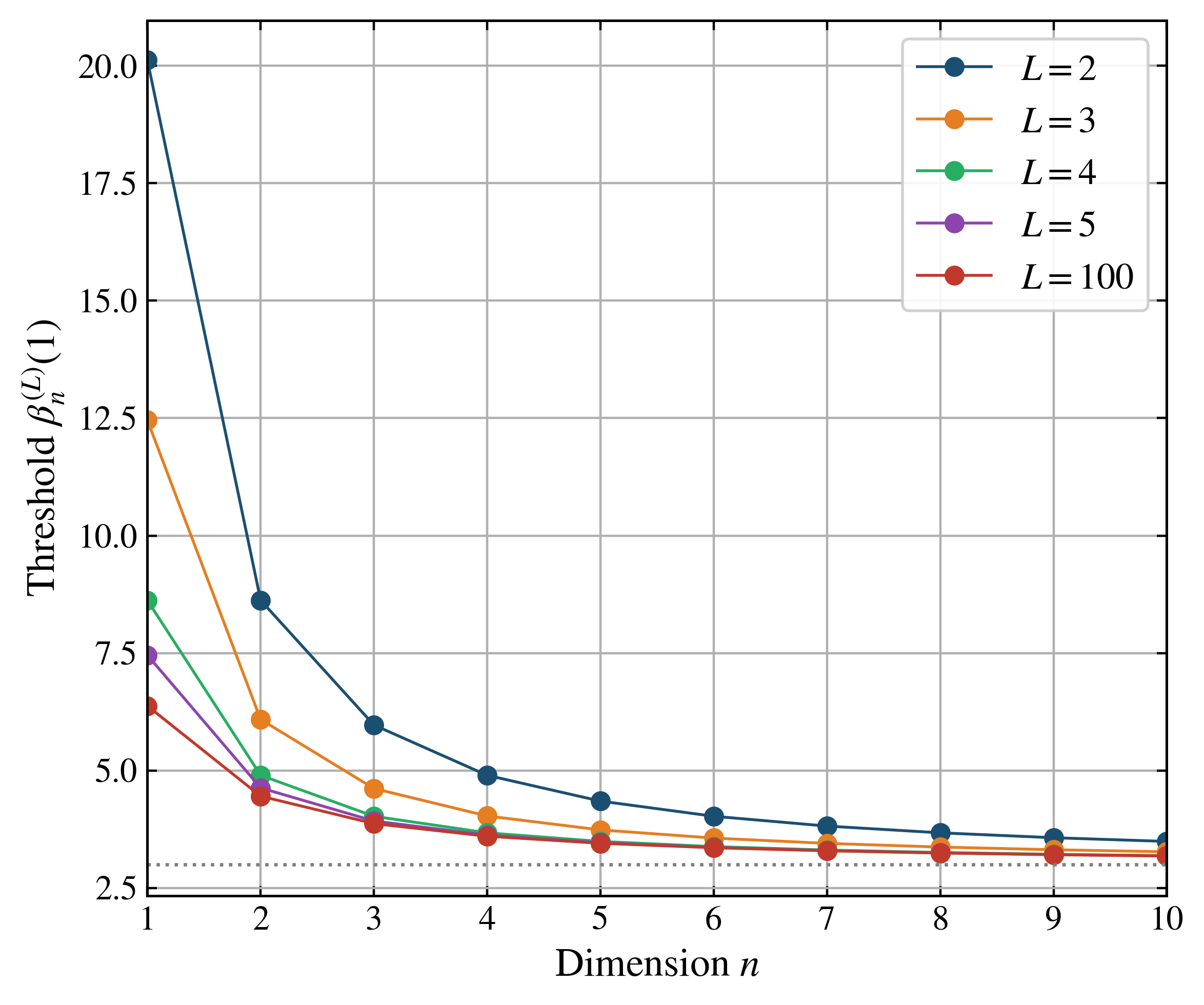}
  \caption{
Threshold \(\beta_n^{(L)}(\alpha)\) as a function of the phenotype-space dimension \(n\), with \(\alpha=1\), \(\mu=0.5\), and \(\nu=1\), for \(L=2,3,4,5,100\). For each value of \(L\), the threshold decreases as \(n\) increases, showing that higher-dimensional phenotype spaces make the abundance of cooperation easier to favor. For the parameter values displayed, increasing the number of phenotypic states \(L\) also lowers the threshold, suggesting a diminishing finite-state effect for larger $L$.
}
  \label{fig3}
\end{figure}

\subsection{\label{sec:phenotype-mutation-rate}Effect of the phenotype mutation rate}
In this section, we examine how the phenotype mutation rate \(\nu\) affects the threshold $\beta_n^{(L)}(\alpha)$. The following result summarizes our findings (the proof is given in Appendix \ref{AppendixD}).

\begin{result}\label{thm:nu-effect-section}
For fixed \(\alpha>0\), \(L\ge2\), \(n\ge1\), and \(\mu>0\), we have
\[
\beta_n^{(L)}(\alpha)\sim\frac{(1+2\mu)(3+2\mu)}{4(1+\mu)\gamma_\alpha}\,\frac1{\nu},
\]
when $\nu\to0^+$, and 
\[
\beta_n^{(L)}(\alpha)\sim\frac{2\,\left(\hatphi(0)\right)^n}{nL^n\,C_n^{(L)}(\alpha)}\,\nu,
\]
when $\nu\to\infty$. Here, 
\[
C_n^{(L)}(\alpha)=\frac{1}{L^n}\sum_{\substack{(p_1,\ldots,p_n)\in\{0,\ldots,L-1\}^n\\ \text{not all }p_j=0}}\frac{\prod_{j=1}^n \hatphi(p_j)}{n-\sum_{j=1}^n \cos(2\pi p_j/L)}.
\]
Consequently, \(\nu\mapsto\beta_n^{(L)}(\alpha)\) is not monotone and attains at least one global minimum at an interior phenotype mutation rate \(\nu^*>0\).
\end{result}

 For \(\nu\) very small, phenotypes change only rarely. This means that individuals can retain the same phenotype for a long time, even when their most recent common ancestor lies far back in the past. Along such long ancestral branches, strategy mutations may occur. Hence, sharing the same phenotype is no longer a reliable signal that two individuals share the same strategy. In this regime, phenotypic discrimination creates little effective assortment and selection requires an increasingly large benefit-to-cost ratio to favor the abundance of cooperation. Increasing \(\nu\) away from this endpoint can lower the threshold because it makes phenotype similarity a more informative signal of recent ancestry and strategic similarity.

A similar problem occurs when \(\nu\) is very large. Since the phenotype space is finite, rapid phenotype mutation makes individuals move quickly through the available phenotype classes. This leads to frequent convergent matches between distantly related individuals. Because these individuals also share long ancestral lineages, phenotypic similarity once again fails to guarantee strategic similarity. Thus, excessive phenotype mutation destroys the assortment that moderate phenotype mutation can generate.

These two limits demonstrate that the critical benefit-to-cost threshold depends non-monotonically on phenotype mutation and has at least one minimum at an intermediate mutation rate. Thus, the abundance of cooperation is favored most easily at an intermediate phenotype mutation rate.

The non-monotone dependence on \(\nu\) is illustrated in Fig.~\ref{fig4}.

\begin{figure}
  \centering
  \includegraphics[height=7cm, width=\columnwidth]{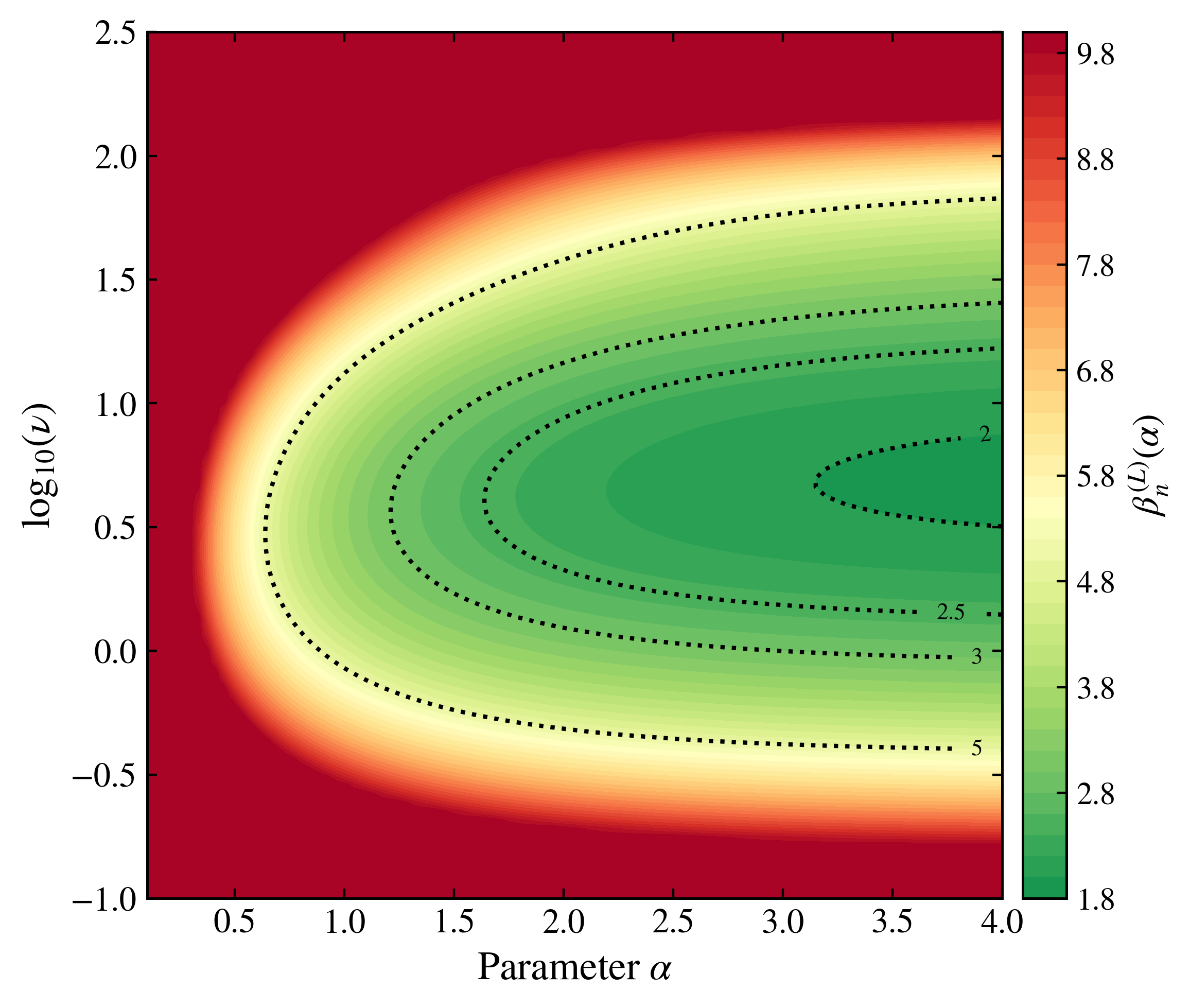}
  \caption{The threshold $\beta_n^{(L)}(\alpha)$ across the two-dimensional parameter space for $L=5$, $n=2$ and $\mu=0.5$. The horizontal axis is $\alpha$, and the vertical axis is the transformed coordinate $\log_{10}(\nu)$. Green indicates a low threshold, at which the abundance of cooperation is easier to favor under weak selection, whereas red indicates a high threshold, at which the abundance of cooperation is harder to favor. The black contour lines show where the threshold equals $2,2.5,3,5$ in the $(\alpha,\log_{10}(\nu))$ plane. The threshold is minimized at an intermediate value of $\nu$. Too little phenotype mutation produces insufficient phenotypic variation for discrimination to generate effective assortment. Too much phenotype mutation mixes phenotypes too rapidly and again destroys the association between phenotypic similarity and ancestry. The low-threshold region forms a band in the $(\alpha,\log_{10}(\nu))$ plane.}
  \label{fig4}
\end{figure}

\subsection{\label{sec:strategy-mutation-rate}Effect of the strategy mutation rate}

Finally, we examine the effect of the strategy mutation rate \(\mu\) on the threshold $\beta_n^{(L)}(\alpha)$. Unlike the phenotype mutation
rate \(\nu\), which has a non-monotone effect, the strategy mutation rate has a one-sided
effect on the threshold given in the following result. The proof is given in Appendix \(\ref{AppendixE}\).

\begin{result}\label{thm:mu-monotonicity-section}
Fix \(\alpha>0\), \(L\ge2\), \(n\ge1\), and \(\nu>0\). Then, the threshold \(\mu\mapsto\beta_n^{(L)}(\alpha)\) is strictly increasing on \((0,\infty)\). Consequently, higher strategy mutation makes selection for the abundance of cooperation more difficult.
\end{result}

This monotonicity has a direct evolutionary interpretation. Increasing \(\mu\) weakens the correspondence between phenotypic similarity and strategic similarity. Even when two individuals with the same or nearby phenotypes have coalesced recently, strategy mutations may have occurred along their ancestral branches. Thus, a cooperator that directs help toward phenotypically similar individuals faces a higher risk of helping a defector. As \(\mu\) becomes large, repeated strategy mutations destroy the association between shared phenotype and shared strategy. Phenotypic similarity may still determine how strongly cooperative behavior is directed toward a partner, but it no longer reliably identifies cooperative partners. Therefore, the effect of increasing \(\mu\) is purely disruptive: it always raises the threshold $\beta_n^{(L)}(\alpha)$.

In the case of large values of \(\mu\), we have
\begin{equation}
\beta_n^{(L)}(\alpha)\sim\frac{2\mathcal L_n^{(L)}\!\left(\alpha,\frac1{2\nu}\right)}{\frac{2\nu}{n}-\mathcal L_n^{(L)}\!\left(\alpha,\frac1{2\nu}\right)}\mu,
\end{equation}
where the coefficient of $\mu$ in the right side is positive (see Corollary \ref{cor:mu_infty_revised}). As a consequence, we obtain
\[
\lim_{\mu\rightarrow\infty}\beta_n^{(L)}(\alpha)=\infty.
\]
This shows that for every fixed benefit-to-cost ratio $b/c$, sufficiently large strategy mutation prevents weak selection from favoring the abundance of cooperation.

These conclusions are illustrated in Fig.~\ref{fig5}.
\begin{figure}
    \centering
    \includegraphics[height=7cm, width=\columnwidth]{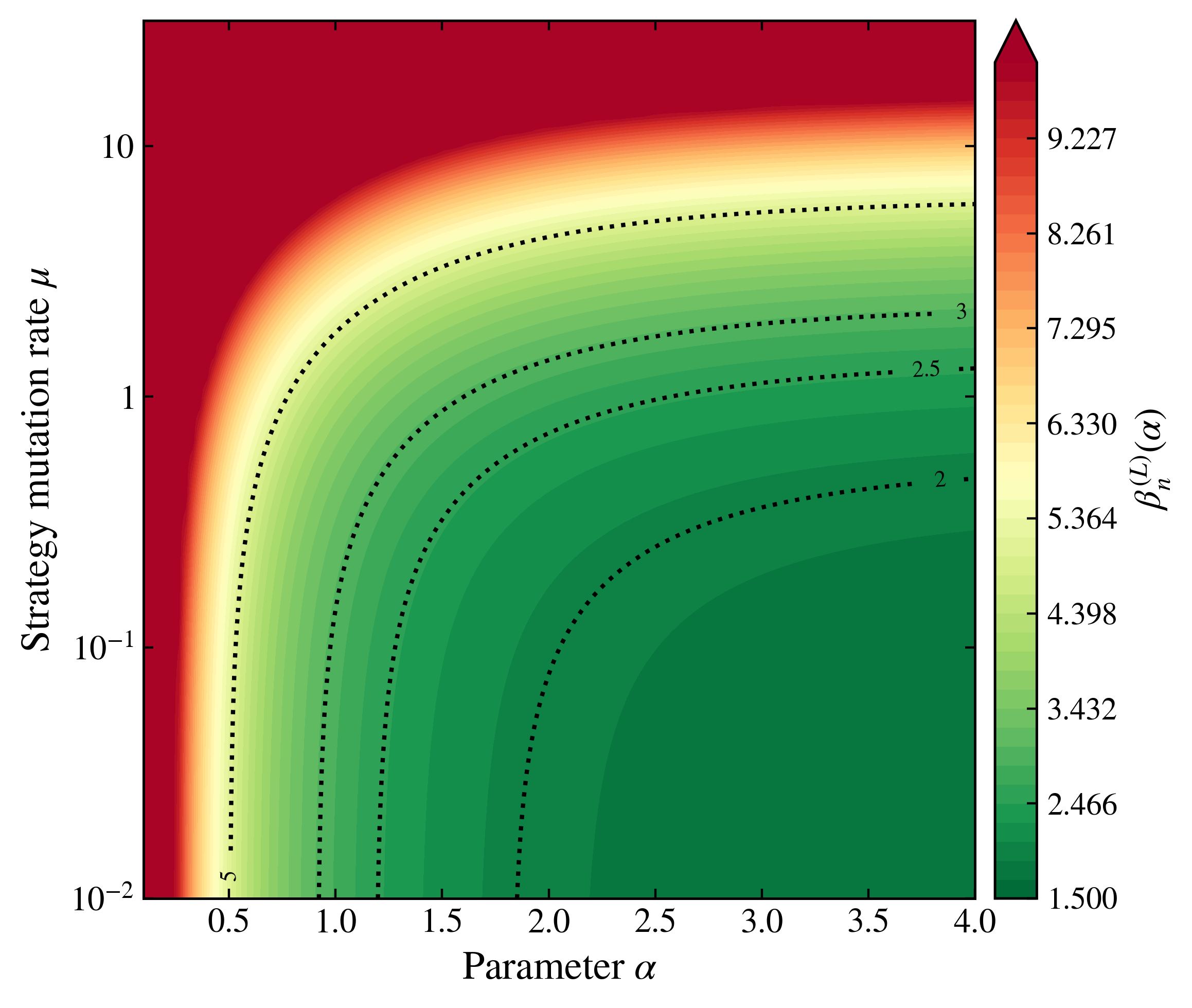}
    \caption{
Effect of parameter $\alpha$ and strategy-mutation rate $\mu$ on the threshold $\beta_n^{(L)}(\alpha)$ with
\(L=5\), \(n=2\), and \(\nu=3\). The vertical axis is shown
on a logarithmic scale. Darker green regions correspond to lower thresholds, where
selection can favor the abundance of cooperation for smaller benefit-to-cost ratios.
Red regions correspond to larger thresholds, where the abundance of cooperation is harder to favor.
The black contour curves indicate the parameter combinations for which
\(\beta_n^{(L)}(\alpha)=2,2.5,3,5\). 
}
    \label{fig5}
\end{figure}

\section{Discussion}

We have studied the abundance of cooperation in a finite population when cooperative behavior is conditioned on phenotypic similarity. Phenotypes lie on the discrete torus \(\mathbb T_L^n\), and a cooperator helps a partner with probability \(e^{-\alpha d_n}\), where \(d_n\) is the periodic phenotypic distance. Under weak selection and the large-population mutation scaling, we obtain an explicit threshold \(\beta_n^{(L)}(\alpha)\) such that weak selection favors the abundance of cooperation to first order precisely when
\[
\frac{b}{c}>\beta_n^{(L)}(\alpha).
\]
The threshold is evaluated using the Kingman coalescent and periodic random walks on the finite torus. In contrast with models on the unbounded lattice \(\mathbb Z^n\) \citep{Antal2009,Kroumi2015,Kroumi2026}, the finite geometry introduces recurrent phenotypic overlap and a uniform long-run distribution of phenotypes.

The unifying theme of the main results is that phenotypic assortment
promotes cooperation only to the extent that phenotypic proximity remains
informative about shared ancestry and strategy. Increasing the discrimination
parameter \(\alpha\) strictly lowers the threshold as this 
reduces helping toward phenotypically distant individuals and concentrates
cooperative benefits on partners whose phenotypes are more strongly associated
with recent common ancestry. In the limit \(\alpha\to\infty\), the interaction
rule becomes exact phenotype matching. By contrast, as
\(\alpha\to0^+\), helping becomes independent of phenotypic distance, the
assortment mechanism disappears, and the cooperation threshold diverges.
This shows that phenotypic variation alone is insufficient; it must be
accompanied by a recognition rule that preferentially directs benefits
toward phenotypically similar individuals. This conclusion is consistent with the graded-recognition
result obtained on \(\mathbb Z^n\) by \citet{Kroumi2026}.

The dimension \(n\) has a monotone effect: increasing the number of phenotypic coordinates strictly
decreases \(\beta_n^{(L)}(\alpha)\). Additional coordinates provide more
independent directions along which two lineages may become phenotypically
distinct. Consequently, remaining close in the full phenotype space becomes a
more selective event and carries more information about common phenotypic
history. This reduces indiscriminate helping toward phenotypically dissimilar
partners and strengthens effective assortment. The result extends the
dimension effects found for exact matching and graded recognition on
unbounded phenotype spaces \citep{Kroumi2015,Kroumi2026}, while showing that
the same mechanism persists when the number of phenotypic states is finite.

The phenotype mutation rate \(\nu\) produces the principal qualitative
difference between finite and unbounded phenotype spaces. On
\(\mathbb Z^n\), increasing phenotype mutation can continue to separate
lineages and thereby strengthen the information contained in phenotypic
similarity \citep{Antal2009,Kroumi2015,Kroumi2026}. On the finite torus,
however, the threshold diverges both as \(\nu\to0^+\) and as
\(\nu\to\infty\). When phenotype mutation is very rare, insufficient
phenotypic variation is generated for discrimination to create effective
assortment. When mutation is very rapid, ancestral lineages repeatedly explore
the entire torus, so identical or nearby phenotypes may arise through
recurrent mixing rather than recent common descent. Phenotypic similarity then
loses much of its genealogical information. The threshold therefore attains
at least one minimum at an intermediate mutation rate \(\nu^*>0\). This
intermediate optimum is a genuinely finite-space effect caused by the
competition between the production of phenotypic variation and the recurrent
mixing of a bounded state space.

Strategy mutation has a different and purely disruptive role. Increasing
\(\mu\) strictly raises the threshold because it weakens the association
between phenotype and behavior. Even individuals with identical or nearby
phenotypes may carry different strategies if strategy mutation has occurred
since their common ancestor. A cooperator that conditions help on phenotypic
similarity is then more likely to direct a benefit toward a defector. This
result gives an analytical form to the concern raised by
\citet{RobertsSherratt2002}: phenotypic similarity becomes ineffective when
the recognition cue does not reliably predict cooperative behavior. It is
also consistent with the reply of \citet{RioloCohenAxelrod2002}, since the
result does not imply that similarity can never support cooperation; rather,
it shows that a larger benefit-to-cost ratio is required as phenotype becomes
less informative about strategy. A related conclusion was obtained by
\citet{Traulsen2007}, although their mutation mechanism couples changes
in strategy and tag, whereas phenotype and strategy mutate independently in
the present model.

The analysis identifies a general principle: finite phenotypic
diversity supports cooperation most effectively when recognition is
sufficiently discriminating, strategy is sufficiently stable, and phenotype
mutation is strong enough to generate variation but not so strong that the
recognition space becomes rapidly mixed.

\appendix

\section*{Data Availability Statement}
The numerical data underlying the figures are generated from the analytical expressions derived in this article. The corresponding numerical data and code will be made publicly available in an appropriate repository.
\section{Derivation of the threshold formula}
\label{app:threshold_derivation}

In this appendix, we provide the technical details leading to the
spectral representation and the identity measures used in
Sec.~\ref{sec:threshold}.

\subsection{Periodic phenotype process}
\subsubsection{\label{app_1:threshold_derivation}Proof of Eq. \eqref{eq:skellam}}
Before coalescence, the two ancestral lineages evolve independently. In coordinate \(r\), each lineage mutates at rate \(\nu/n\). At each such mutation, the coordinate changes by \(+1\) or \(-1\) modulo \(L\), each with probability \(1/2\). Therefore, the difference $Z_r(s)=y_r(k,s)-y_r(j,s)\pmod L$ jumps by \(+1\) either when the first lineage mutates by \(+1\) or when the second lineage mutates by \(-1\). Hence the rate of \(+1\) jumps is $\nu/(2n)+\nu/(2n)=\nu/n$. Similarly, the rate of \(-1\) jumps is also \(\nu/n\). Thus \(Z_r\) is a continuous-time symmetric nearest-neighbor random walk on \(\mathbb T_L=\mathbb Z/L\mathbb Z\) with total jump rate \(2\nu/n\).

Therefore, the generator matrix of \(Z_r\) is 
\[
Q=\frac{2\nu}{n}(A-I),
\]
where the matrix 
\(A=(A(x,y))_{x,y\in\mathbb T_L}\) is given by
\[
A(x,y)=\frac12\mathbf 1_{\{y=x+1\pmod L\}}+\frac12\mathbf 1_{\{y=x-1\pmod L\}}.
\]
Conditional on \(T_{2}=\tau\), the process starts from
\(Z_r(0)=0\), and its transition matrix after time \(\tau\) is $e^{\tau Q}$.
Hence,
\begin{align*}
&\PP_0\bigl(
y_r(k)-y_r(j)=m\pmod L
\mid T_{2}=\tau
\bigr)\\
&\qquad=
\bigl(e^{-\sigma(I-A)}\bigr)_{0m}.
\end{align*}
where  \(\sigma=2\nu\tau/n\). Since \(I\) and \(A\) commute, we obtain
\[
e^{-\sigma(I-A)}=e^{-\sigma}e^{\sigma A}.
\]
It therefore remains to compute \(e^{\sigma A}\). For \(p=0,\ldots,L-1\), set $\zeta_p=e^{-2\pi\mathrm{i}p/L}$ and define $v_p= (1,\zeta_p, \zeta_p^2,\cdots,\zeta_p^{L-1})^{\!T}$.
A direct matrix multiplication gives
\[
Av_p=\frac12\left(\zeta_p+\zeta_p^{-1}\right)v_p=\cos\left(\frac{2\pi p}{L}\right)v_p=c_pv_p.
\]
Thus, \(v_p\) is an eigenvector of \(A\) associated with the eigenvalue $c_p$. Moreover, we have
\[
v_p^*v_q=\sum_{x=0}^{L-1}e^{2\pi\mathrm{i}(p-q)x/L}=
\begin{cases}
L, & p=q,\\
0, & p\ne q,
\end{cases}
\]
where the second equality follows from the finite geometric-series
formula. Hence, \(v_0,\ldots,v_{L-1}\) form an orthogonal basis of
\(\mathbb C^L\). If
\[
V=(v_0\ v_1\ \cdots\ v_{L-1}),
\]
then
\[
V^{-1}=\frac1L V^*
\]
and
\[
A=V\operatorname{diag}(c_0,\ldots,c_{L-1})V^{-1}.
\]
Consequently,
\[e^{\sigma A}=V\operatorname{diag}\left(e^{\sigma c_0},\ldots,e^{\sigma c_{L-1}}\right)V^{-1}=\frac1L\sum_{p=0}^{L-1}e^{\sigma c_p}v_pv_p^*.\]
Therefore,
\[
e^{-\sigma(I-A)}=\frac1L\sum_{p=0}^{L-1}e^{-\sigma(1-c_p)}v_pv_p^*.
\]
Taking the \((0,m)\)-entry and using
\[
v_p(0)=1,\qquad\overline{v_p(m)}=e^{2\pi\mathrm{i}pm/L},
\]
we obtain
\[
\bigl(e^{-\sigma(I-A)}\bigr)_{0m}=\frac1L\sum_{p=0}^{L-1}e^{-\sigma(1-c_p)}e^{2\pi\mathrm{i}pm/L}.
\]
Thus,
\[
\begin{aligned}
&\PP_0\bigl(y_r(k)-y_r(j)=m\pmod L\mid T_{2}=\tau\bigr)\\
&=\frac1L\sum_{p=0}^{L-1}e^{-\sigma(1-c_p)}e^{2\pi\mathrm{i}pm/L},
\end{aligned}
\]
which proves \eqref{eq:skellam}.

\subsubsection{\label{app_2:threshold_derivation}Expression of $\hatphi(p)$ and Properties of $\Phic(\sigma)$}
Let us start with the expression of $\hatphi(p)$.
If $L=2M$ is even, the distances are $0, 1, \dots, M-1, M, M-1, \dots, 1$, where $m=M$ is the unique point at distance $M$ and does not have a symmetric pair. Its contribution to the sum is
\begin{equation*}
e^{-\alpha M} e^{-2\pi \mathrm{i} p M / (2M)}  = (-1)^pe^{-\alpha M} .
\end{equation*}
Summing the contribution of the origin, the pairs from $1$ to $M-1$, and the unique midpoint, we obtain
\begin{equation*}
\hatphi(p) = 1 + 2\sum_{m=1}^{M-1} e^{-\alpha m} \cos\left(\frac{2\pi p m}{L}\right) + (-1)^pe^{-\alpha M}.
\end{equation*}

If $L=2M+1$ is odd, then every non-zero distance $m \in \{1, \dots, M\}$ appears exactly twice in the set $d_1(0,1), \dots, d_1(0, L-1)$, where $d_1(0,m) = d_1(0, L-m) = m$ for $1 \le m \le M$. Now, grouping the terms for $m$ and $L-m$, we obtain
\begin{align}
\hatphi(p) &= 1 + \sum_{m=1}^{M}e^{-\alpha m} \left(  e^{-2\pi \mathrm{i} p m/L} +  e^{-2\pi \mathrm{i} p (L-m)/L} \right)\nonumber \\
&= 1 + 2\sum_{m=1}^{M} e^{-\alpha m} \cos\left(\frac{2\pi p m}{L}\right).
\end{align}

For the properties of $\Phic(\sigma)$, we have the following Lemma.
\begin{lemma}\label{lem:phi_properties}
For every $\alpha>0$ and every $L\ge2$, the following statements hold.
\begin{enumerate}
\item [(i)] $\hatphi(p)>0$ for any $p=0,\ldots,L-1$.
\item [(ii)] $\Phic(0)=1$.
\item [(iii)] $\lim_{\sigma\to\infty}\Phic(\sigma)= \hatphi(0)/L$.
\item [(iv)] $\partial_\sigma \Phic(0)=-\gamma_\alpha<0$.
\item [(v)] $\partial_\sigma \Phic(\sigma)<0$, for any $\sigma>0$.
\end{enumerate}
\end{lemma}

\begin{proof}
Let $r=e^{-\alpha}\in(0,1)$ and $\theta=\frac{2p\pi}{L}$.

We start with part $(i)$. If $L=2M$ is even, we have
\begin{align*}
\hatphi(p) 
=&1 + (-1)^pr^{M}+ 2\sum_{m=1}^{M-1} r^{m} \cos\left(\theta m\right) \\
=& 1 + (-1)^pr^{M}+2\operatorname{Re} \left[ \sum_{m=1}^{M-1} (re^{\mathrm{i}\theta})^m \right] \\
=& 1 + (-1)^pr^{M}+2\operatorname{Re} \left[  \frac{re^{\mathrm{i}\theta} - (re^{\mathrm{i}\theta})^M}{1 - re^{\mathrm{i}\theta}} \right] \\
=&1 + (-1)^pr^{M}\\
&+ 2 \frac{r\cos\theta+ (-1)^pr^{M+1}\cos\theta -r^2  - (-1)^pr^{M} }{1-2r\cos(\theta)+r^2}  \\
=& \frac{(1-r^2)\bigl(1-(-1)^p r^M\bigr)}{1-2r\cos(\theta)+r^2}.
\end{align*}
The denominator is strictly positive because $0<r<1$, and the numerator is also strictly positive. Hence $\hatphi(p)>0$.

If $L=2M+1$ is odd, the same computation gives
\[
\hatphi(p)=\frac{(1-r)\Bigl(1+r-2(-1)^p r^{M+1}\cos(\theta/2)\Bigr)}{1-2r\cos(\theta)+r^2}.
\]
Again the denominator is strictly positive. Also, note that
\[
1+r-2(-1)^p r^{M+1}\cos(\theta/2)\ge 1-r^{M+1}+r(1-r^{M})>0,
\]
so $\hatphi(p)>0$ in the odd case as well.

For part $(ii)$, let
\[
K(m)=e^{-\alpha d_1(0,m)},\qquad m\in\mathbb T_L,
\]
and let \((X_\sigma)_{\sigma\ge0}\) be the rate-one continuous-time symmetric
random walk on \(\mathbb T_L\) starting from \(0\). Note that the function \(\Phic\) can be written as
\begin{equation}\label{eq_main}
\Phic(\sigma)
=
\E_0\left[K(X_\sigma)\right].
\end{equation}
Indeed, \(X_\sigma\) has the same distribution as the one-coordinate
phenotypic difference conditional on \(T_{2}=\tau\), with
\(\sigma=2\nu\tau/n\).
At \(\sigma=0\), we have \(X_0=0\) almost surely. Therefore,
\[
\Phic(0)=K(0)=e^{-\alpha d_1(0,0)}=1.
\]

For part $(iii)$, only the mode $p=0$ survives when $\sigma\to\infty$, because $c_0=1$ and $c_p<1$ for every $p\ne0$. Therefore, we obtain
\[
\lim_{\sigma\to\infty}\Phic(\sigma)= \frac{\hatphi(0)}{L}.
\]

For part $(iv)$, using the representation \eqref{eq_main}, consider a small time interval \(h>0\). Since \(X_\sigma\) is a rate-one continuous-time symmetric random walk starting from \(0\), as \(h\downarrow0\) we have
\[
\PP_0(X_h=0)=1-h+o(h),
\]
while a single jump occurs to either neighboring site with probability \(h/2+o(h)\). The probability of two or more jumps is \(o(h)\). Therefore,
\begin{equation*}
\begin{split}
\Phic(h)&=\E_0[K(X_h)]\\
&=(1-h)K(0)+\frac{h}{2}K(1)+\frac{h}{2}K(-1)+o(h).
\end{split}
\end{equation*}
Since \(K(0)=1\) and \(K(1)=K(-1)=e^{-\alpha}\), we obtain
\[
\Phic(h)=1-h(1-e^{-\alpha})+o(h).
\]
Hence,
\[
\partial_\sigma\Phic(0)=\lim_{h\downarrow0}\frac{\Phic(h)-\Phic(0)}{h}=-(1-e^{-\alpha}).
\]

For part (\textit{v}), differentiating \eqref{eq:Phidef} gives
\[
\partial_\sigma \Phi^{(L)}_\alpha(\sigma)
=
-\frac{1}{L}\sum_{p=0}^{L-1}(1-c_p)\,\hatphi(p)\,e^{-\sigma(1-c_p)}<0,
\]
since all $\hatphi(p)>0$ by part~$(i)$ and $(1-c_p)>0$ for $p\ne0$.
\end{proof}

\subsection{\label{app_identity_measures}Identity measures}
\subsubsection{Two-lineage identity measures}
For  $\ZnL(\alpha)$, by conditioning on $T_2$ the coalescence time of individuals $k$ and $j$ and using \eqref{eqqq}, we obtain
\begin{align*}
\ZnL(\alpha)&=\E_0\left[\varphi(k,j)\right]\\
&=\int_0^\infty\E_0\left[\varphi(k,j)\Big| T_2=\tau\right]e^{-\tau}\,d\tau\\
&=\int_0^\infty\left[\Phic\left(\frac{2\nu\tau}{n}\right)\right]^ne^{-\tau}\,d\tau\\
&=\frac{n}{2\nu}\int_0^\infty[\Phic(\sigma)]^n e^{-n\sigma/(2\nu)}\,d\sigma\\
&=\frac{n}{2\nu}\,\LnL\left(\alpha,\frac{1}{2\nu}\right).
\end{align*}

For $\GnL(\alpha)$, again condition on the coalescence time $T_2$ and use \eqref{eq:strategyprob} to get
\begin{align*}
&\GnL(\alpha)\\
&=\E_0\left[\varphi(k,j)\mathbf 1_{\{S(k)=S(j)\}}\right]\\
&=\int_0^\infty\E_0\bigl[\varphi(k,j)\mathbf 1_{\{S(k)=S(j)\}}|T_2=\tau\bigr]e^{-\tau}\,d\tau\\
&=\int_0^\infty\E_0\bigl[\varphi(k,j)|T_2=\tau\bigr]\PP_0\bigl(S(k)=S(j)\mid T_2=\tau\bigr)e^{-\tau}\,d\tau\\
&=\int_0^\infty\left[\Phic\left(\frac{2\nu\tau}{n}\right)\right]^n\frac{1+e^{-2\mu\tau}}{2}e^{-\tau}\,d\tau\\
&=\frac{n}{4\nu}\int_0^\infty[\Phic(\sigma)]^n\left[e^{-n\sigma/(2\nu)}+e^{-n(1+2\mu)\sigma/(2\nu)}\right]\,d\sigma\\
&=\frac{n}{4\nu}\left[\LnL\left(\alpha,\frac{1}{2\nu}\right)+\LnL\left(\alpha,\frac{1+2\mu}{2\nu}\right)\right].
\end{align*}

\subsubsection{Three-lineage identity measure}

Let $T_{3}$ denote the time until the first coalescence among the three sampled lineages of individuals $k$, $j$ and $\ell$, and let $T_{2}$ denote the additional time until the remaining two lineages coalesce. The joint density of $(T_{2},T_{3})$ is
\[
f(\tau_2,\tau_3)=3e^{-(3\tau_3+\tau_2)},
\]
for $\tau_3,\tau_2>0$.
Each of the three possible first-coalescing pairs among $(k,j)$, $(j,\ell)$ and $(k,\ell)$  occurs with probability $1/3$.
Then, we have
\begin{equation}\label{eq1}
\HnL(\alpha)=\frac13\Bigl(H_n^{(1)}(\alpha)+H_n^{(2)}(\alpha)+H_n^{(3)}(\alpha)\Bigr),
\end{equation}
where the three terms correspond to the three possible first coalescence events $(k,j)$, $(j,\ell)$ and $(k,\ell)$, respectively.

If $k$ and $j$ coalesce first at time $\tau_3$, then the phenotype factor depends on $\tau_3$, while the strategy factor depends on the total time $\tau_3+\tau_2$. Hence, we have
\begin{equation}\label{eq2}
\begin{split}
&H_n^{(1)}(\alpha)\\
=&\int_0^\infty\int_0^\infty3e^{-(3\tau_3+\tau_2)}\left[\Phic\left(\frac{2\nu\tau_3}{n}\right)\right]^n\\
&\qquad\qquad\qquad\qquad\qquad \times\frac{1+e^{-2\mu(\tau_3+\tau_2)}}{2}\,d\tau_2\,d\tau_3\\
=&\frac32\int_0^\infty\left[\Phic\left(\frac{2\nu\tau}{n}\right)\right]^n\left(e^{-3\tau}+\frac{1}{1+2\mu}e^{-(3+2\mu)\tau}\right)d\tau.
\end{split}
\end{equation}

If $j$ and $\ell$ coalesce first at time $\tau_3$, then the strategy factor depends on $\tau_3$, while the phenotype factor depends on the total time $\tau_3+\tau_2$. Therefore, we have
\begin{align}
H_n^{(2)}(\alpha)&=\int_0^\infty\int_0^\infty3e^{-(3\tau_3+\tau_2)}\left[\Phic\left(\frac{2\nu(\tau_3+\tau_2)}{n}\right)\right]^n\nonumber\\
&\qquad\qquad \qquad \qquad \qquad \times\frac{1+e^{-2\mu\tau_3}}{2}\,d\tau_2\,d\tau_3\nonumber\\
&=\frac34\int_0^\infty\left[\Phic\left(\frac{2\nu\tau}{n}\right)\right]^n\Bigg(\frac{2+\mu}{1+\mu}e^{-\tau}-e^{-3\tau}\nonumber\\
&\qquad\qquad\qquad\qquad-\frac{1}{1+\mu}e^{-(3+2\mu)\tau}\Bigg)d\tau.\label{eq3}
\end{align}

If $k$ and $\ell$ coalesce first, then both the phenotype factor and the strategy factor depend on the total time $\tau_3+\tau_2$. Thus
\begin{align}
H_n^{(3)}(\alpha)&=\frac34\int_0^\infty\left[\Phic\left(\frac{2\nu\tau}{n}\right)\right]^n\Big(e^{-\tau}+e^{-(1+2\mu)\tau}\nonumber\\
&\qquad\qquad\qquad-e^{-3\tau}-e^{-(3+2\mu)\tau}\Big)d\tau.\label{eq4}
\end{align}

Substituting Eqs.~\eqref{eq2}--\eqref{eq4} into \eqref{eq1} leads to Eq.~\eqref{eq:Hexplicit}.

\subsection{\label{ln_is_decreasing}Generalized Laplace transform}
Substituting the spectral representation \eqref{eq:Phidef} into \eqref{eq:Ldef} gives
\begin{align*}
\LnL(\alpha,x)
&=
\frac1{L^n}
\sum_{p_1,\ldots,p_n=0}^{L-1}
\left(\prod_{j=1}^n\hatphi(p_j)\right)\\
&\qquad\times
\int_0^\infty
e^{-\sigma\left(
n-\sum_{j=1}^n c_{p_j}+nx
\right)}
\,d\sigma.
\end{align*}
Since $c_{p_j}\le1$, we have
\[
n-\sum_{j=1}^n c_{p_j}+nx\ge nx>0,
\]
so each integral converges and we have 
\[
\int_0^\infty
e^{-\sigma\left(
n-\sum_{j=1}^n c_{p_j}+nx
\right)}
\,d\sigma=\frac1{n-\sum_{j=1}^n c_{p_j}+nx}.
\]
This proves that 
\begin{equation}\label{eq:closed_form}
\LnL(\alpha,x)=\frac{1}{L^n}\sum_{p_1,\ldots,p_n=0}^{L-1}\frac{\prod_{j=1}^n \hatphi(p_j)}{n-\sum_{j=1}^n c_{p_j}+nx}.
\end{equation}
The same formula combined with Lemma \ref{lem:phi_properties} (i) shows that $x\mapsto\LnL(\alpha,x)$ decreases strictly on $(0,\infty)$.

\section{\label{AppendixB}Effect of the Discrimination Parameter}
In this section, we will study the threshold $\beta_n^{(L)}(\alpha)$ with respect to the discrimination parameter $\alpha>0$. The rest of this section is organized as follows. First, we recall the definition of total positivity. Then, Lemma~\ref{lem:time_state_TP2} establishes the required time--state total positivity property for the one-coordinate distance process. Lemma~\ref{lem:mlr_order} shows that this property is preserved under positive tilting and derives the associated stochastic-ordering consequences. Using these two results, Lemma~\ref{lem:cross} proves that the mixed derivative of the logarithm of the one-coordinate transform with respect to time and the discrimination parameter is strictly negative. Lemma~\ref{lem:X_monotone} then transfers this pointwise monotonicity to the generalized Laplace transform appearing in the threshold formula. Theorem~\ref{thm:alpha} combines these monotonicity properties with the explicit representation of the threshold to prove that \(\alpha\mapsto\beta_n^{(L)}(\alpha)\) is strictly decreasing. Finally, Corollary~\ref{cor:alpha} identifies the limiting behavior of the threshold as \(\alpha\to0^+\) and \(\alpha\to\infty\).

\begin{definition}
\begin{enumerate}
\item  [(i)] A matrix \(A=(a_{ij})\) is called \emph{totally positive} if every
minor of \(A\) is nonnegative. Equivalently, for every integer \(r\ge 1\)
and every choice of indices $i_1<\cdots<i_r$ and  $j_1<\cdots<j_r$,
we have
\[
\det\bigl(a_{i_pj_q}\bigr)_{p,q=1}^r \ge 0.
\]
\item  [(ii)] A positive kernel \((x,y)\mapsto K(x,y)\) on two totally ordered sets is called \emph{totally positive of order two} (TP2) if for every \(x_1<x_2\) and \(y_1<y_2\), we have
\[
K(x_1,y_2)K(x_2,y_1)\le K(x_1,y_1)K(x_2,y_2).
\]
Equivalently, for every \(x_1<x_2\), the ratio
\(
y\longmapsto \dfrac{K(x_2,y)}{K(x_1,y)}
\)
is increasing.
\end{enumerate}
\end{definition}

\begin{lemma}\label{lem:time_state_TP2}
Let $(X_\sigma)_{\sigma\ge0}$ be the one-coordinate continuous-time simple random walk on $\mathbb{T}_L$ started from $0$ and 
$\mathcal{D}_\sigma=d_1(0,X_\sigma)\in\{0,1,\ldots,M\}$ be the one-coordinate distance process, with
$M=\lfloor L/2\rfloor$.
Let
\[
\mathcal{P}_\sigma(d)=\PP_0(\mathcal{D}_\sigma=d\mid \mathcal{D}_0=0),
\]
for $d=0,1,\ldots,M$. Then the kernel $(\sigma,d)\longmapsto \mathcal{P}_\sigma(d)$
is totally positive of order two on $(0,\infty)\times\{0,1,\ldots,M\}$. 
\end{lemma}

\begin{proof}
Note that the process $(\mathcal{D}_\sigma)_{\sigma\ge0}$ is a finite birth--death process on $\{0,1,\ldots,M\}$.
Let $\mathcal{A}$ be its generator, which is a finite Jacobi matrix
with nonnegative off-diagonal entries. Therefore, by 
\cite[Chapter~3, Theorem~3.4(a), p.~115]{Karlin1968},
the matrix $\mathcal P_h=e^{h\mathcal{A}}$ is totally positive for every $h\ge0$.

Now, fix $h>0$ and consider $\mathcal P_h$ as the one-step transition matrix of a discrete-time Markov chain on the ordered state space $\{0,1,\ldots,M\}$. By \cite[Chapter~1, Theorem~7.2, p.~43]{Karlin1968}, the $m$-step transition function
\[
(m,d)\longmapsto \mathcal P_h^m(0,d)=\mathcal P_{mh}(0,d)=\mathcal P_{mh}(d)
\]
is totally positive, in particular of order two. Here $\mathcal P_h^m(0,d)$ is the entry $(0,d)$ of the $m$-step transition matrix 
\[
\mathcal P^{m}_h=(e^{h\mathcal{A}})^m=e^{mh\mathcal{A}}=\mathcal P_{hm}.
\]
Then, for all integers \(m_1<m_2\) and all \(d_1<d_2\), we have
\begin{equation}\label{eq:tp2_discrete_kernel}
\mathcal P_{m_1h}(d_1)\mathcal P_{m_2h}(d_2)\ge \mathcal P_{m_1h}(d_2)\mathcal P_{m_2h}(d_1).
\end{equation}

Finally, let $0<\sigma_1<\sigma_2$ be two positive times and $d_1<d_2$ be two integers. For each $n\ge1$, set $h_n=\frac{1}{n}$ and $m_r^{(n)}=\lfloor n\sigma_r\rfloor$, for $r=1,2$. For all sufficiently large $n$, we have
\[
m_1^{(n)}<m_2^{(n)}.
\]
Using Eq.~\eqref{eq:tp2_discrete_kernel}, we obtain
\[
\mathcal P_{m_1^{(n)}h_n}(d_1)\,\mathcal P_{m_2^{(n)}h_n}(d_2)
\ge
\mathcal P_{m_1^{(n)}h_n}(d_2)\,\mathcal P_{m_2^{(n)}h_n}(d_1).
\]
Since the matrix exponential \(\mathcal P(t)=e^{t\mathcal{A}}\) depends continuously on \(t\), each entry \(\mathcal P_t(d)\) is continuous in \(t\). Letting \(n\to\infty\) and using
\[
\lim_{n\to\infty}m_r^{(n)}h_n=\sigma_r,
\]
we obtain
\[
\mathcal P_{\sigma_1}(d_1)\mathcal P_{\sigma_2}(d_2)\ge \mathcal P_{\sigma_1}(d_2)\mathcal P_{\sigma_2}(d_1).
\]
This proves that the kernel $(\sigma,d)\mapsto \mathcal P_\sigma(d)$ is totally positive of order two.
\end{proof}

The next lemma is a simple permanence property of monotone likelihood-ratio order under positive tilting.

\begin{lemma}\label{lem:mlr_order}
Let \(\{\kappa_t\}_{t>0}\) be a family of strictly positive probability distributions on the ordered set \(\{0,1,\ldots,M\}\), such that the kernel
$(t,d)\longmapsto \kappa_t(d)$
is totally positive of order two. Let $w(d)>0$ for all $d\in\{0,1,\ldots,M\}$, and define a new family of distributions
\[
\widetilde{\kappa}_t(d)
=
\frac{w(d)\kappa_t(d)}{\sum_{m=0}^M w(m)\kappa_t(m)}.
\]
Then the following statements hold.

\begin{enumerate}
\item The kernel $(t,d)\longmapsto \widetilde{\kappa}_t(d)$ is totally positive of order two.

\item For every $0<t_1<t_2$, the distribution $\widetilde{\kappa}_{t_2}$ stochastically dominates $\widetilde{\kappa}_{t_1}$, that is to say \[
\sum_{d=m}^M\bigl(\widetilde{\kappa}_{t_2}(d)-\widetilde{\kappa}_{t_1}(d)\bigr)\ge0,
\]
for any $m\in\{0,1,\ldots,M\}$.

\item For every increasing function \(\psi\), the map
\[
t\longmapsto \sum_{d=0}^M\psi(d)\widetilde{\kappa}_t(d)
\]
is nondecreasing. Moreover, if \(0<t_1<t_2\),
\(\widetilde{\kappa}_{t_1}\neq\widetilde{\kappa}_{t_2}\), and \(\psi\) is
strictly increasing, then
\[
\sum_{d=0}^M\psi(d)\widetilde{\kappa}_{t_2}(d)
>
\sum_{d=0}^M\psi(d)\widetilde{\kappa}_{t_1}(d).
\]
\end{enumerate}
\end{lemma}

\begin{proof}
For part (i), fix $0<t_1<t_2$ and $d_1<d_2$. As $(t,d)\mapsto \kappa_t(d)$ is $TP_2$, we have
\[
\begin{aligned}
&\widetilde{\kappa}_{t_1}(d_1)\widetilde{\kappa}_{t_2}(d_2)-\widetilde{\kappa}_{t_1}(d_2)\widetilde{\kappa}_{t_2}(d_1)\\
=&
\frac{w(d_1)w(d_2)}{Z_{t_1}Z_{t_2}}
\Bigl(
\kappa_{t_1}(d_1)\kappa_{t_2}(d_2)
-
\kappa_{t_1}(d_2)\kappa_{t_2}(d_1)
	\Bigr)\ge 0,
\end{aligned}
\]
where $Z_t=\sum_{m=0}^M w(m)\kappa_t(m)$. This shows that \((t,d)\mapsto \widetilde{\kappa}_t(d)\) is also $TP_2$.

For part (ii), fix $0<t_1<t_2$ and define
\[
r(d)=\frac{\widetilde{\kappa}_{t_2}(d)}{\widetilde{\kappa}_{t_1}(d)}.
\]
By (i), the function $d\mapsto r(d)$ is non-decreasing. On the other hand, we have
\[
\sum_{d=0}^M \widetilde{\kappa}_{t_1}(d)\bigl(r(d)-1\bigr) =\sum_{d=0}^M \bigl(\widetilde{\kappa}_{t_2}(d)-\widetilde{\kappa}_{t_1}(d)\bigr)=0.
\]
Then, there exists $d_0\in\{0,\ldots,M\}$ such that
\[
r(d)\le1\quad\text{for }d\le d_0,\qquad r(d)\ge1\quad\text{for }d>d_0.
\]
Set
\[
\Delta(d)=\widetilde{\kappa}_{t_2}(d)-\widetilde{\kappa}_{t_1}(d)=\widetilde{\kappa}_{t_1}(d)\bigl(r(d)-1\bigr).
\]
Since \(\widetilde{\kappa}_{t_1}(d)>0\), we have
\[
\Delta(d)\leq0
\quad\text{for }d\leq d_0,
\qquad
\Delta(d)\geq0
\quad\text{for }d>d_0.
\]
Now fix \(m\in\{0,\ldots,M\}\). If \(m>d_0\), then
\[
\sum_{d=m}^M\Delta(d)\geq0
\]
because every term in the sum is nonnegative. If \(m\leq d_0\), then,
using \(\sum_{d=0}^M\Delta(d)=0\), we obtain
\[
\sum_{d=m}^M\Delta(d)
=
-\sum_{d=0}^{m-1}\Delta(d)\geq0,
\]
because \(\Delta(d)\leq0\) for \(d\leq m-1\). Therefore, we have
\[
\sum_{d=m}^M
\left(
\widetilde{\kappa}_{t_2}(d)
-
\widetilde{\kappa}_{t_1}(d)
\right)
\geq0
\]
for every \(m\in\{0,\ldots,M\}\).

For part (iii), define
\[
K_t(m)=\sum_{d=m}^M\widetilde{\kappa}_t(d),
\qquad m=1,\ldots,M.
\]
By part (ii), we have
\[
K_{t_2}(m)\geq K_{t_1}(m)
\]
for every \(m\). Let \(\psi:\{0,\ldots,M\}\to\mathbb R\) be increasing. Since
\[
\psi(d)
=
\psi(0)+\sum_{m=1}^d\bigl(\psi(m)-\psi(m-1)\bigr),
\]
we obtain
\[
\sum_{d=0}^M\psi(d)\widetilde{\kappa}_t(d)
=
\psi(0)
+
\sum_{m=1}^M
\bigl(\psi(m)-\psi(m-1)\bigr)K_t(m).
\]
Therefore, we deduce 
\begin{align*}
&\sum_{d=0}^M\psi(d)\left(\widetilde{\kappa}_{t_2}(d)-\widetilde{\kappa}_{t_1}(d)\right)\\
=&
\sum_{m=1}^M
\bigl(\psi(m)-\psi(m-1)\bigr)
\bigl(K_{t_2}(m)-K_{t_1}(m)\bigr)
\geq0,
\end{align*}
because both factors in every term are nonnegative. This shows that
\[
\sum_{d=0}^M\psi(d)\widetilde{\kappa}_{t_2}(d)
\geq
\sum_{d=0}^M\psi(d)\widetilde{\kappa}_{t_1}(d).
\]
The inequality is strict whenever, for some \(m\), we have
\[
\psi(m)>\psi(m-1)
\quad\text{and}\quad
K_{t_2}(m)>K_{t_1}(m).
\]

\end{proof}

\begin{lemma}\label{lem:cross}
For every \(\alpha>0\), the map
\[
\sigma\longmapsto
\frac{\partial}{\partial\alpha}
\log\Phic(\sigma)
\]
is strictly decreasing on \((0,\infty)\).
\end{lemma}

\begin{proof}
By Eq.~\eqref{eq_main}, we have
\[
\Phic(\sigma)=\E_0\left[e^{-\alpha \mathcal{D}_\sigma}\right],
\]
where $\mathcal{D}_\sigma$ is the one-coordinate distance process defined in Lemma \ref{lem:time_state_TP2}. 
Then, we obtain
\begin{equation}\label{eq_200}
\begin{split}
\frac{\partial}{\partial\alpha}\log\Phic(\sigma)=\frac{\frac{\partial}{\partial\alpha}\Phic(\sigma)}{\Phic(\sigma)}&=-\frac{\E_0\left[\mathcal{D}_\sigma e^{-\alpha \mathcal{D}_\sigma} \right]}{\Phic(\sigma)}\\
&=
-\E_{\PP^{(L)}_{\sigma,\alpha}}[\mathcal{D}_\sigma],
\end{split}
\end{equation}
where $\PP^{(L)}_{\sigma,\alpha}$ is the discrete probability measure on $\{0,1,\ldots,M\}$ given by 
{
\[
\PP^{(L)}_{\sigma,\alpha}(d)=\dfrac{e^{-\alpha d}\mathcal P_\sigma(d)}{\Phic(\sigma)}.
\]
}
Using Lemma \ref{lem:time_state_TP2}, the kernel $(\sigma,d)\longmapsto \mathcal P_\sigma(d)$ is totally positive of order two.
Then, Lemma \ref{lem:mlr_order}(i) shows that the tilted kernel
\[
(\sigma,d)\longmapsto \PP^{(L)}_{\sigma,\alpha}(d)
\]
is also totally positive of order two. Hence, by Lemma \ref{lem:mlr_order}(iii) with $\psi(d)=d$, the map
\[
\sigma\longmapsto \E_{\PP^{(L)}_{\sigma,\alpha}}[\mathcal{D}_\sigma]
\]
is strictly increasing on $(0,\infty)$, which completes the proof.
\end{proof}

The next lemma turns this monotonicity into a monotonicity result for the generalized Laplace transform.

\begin{lemma}\label{lem:X_monotone}
For $a>0$, define
\[
\Theta(a)=\frac{\partial}{\partial\alpha}\log\LnL(\alpha,a).
\]
Then $\Theta(a)<0$ for every $a>0$, and the map $a\mapsto \Theta(a)$ is strictly increasing on $(0,\infty)$.
\end{lemma}

\begin{proof}
 Differentiating \eqref{eq:Ldef} with respect to \(\alpha\), we obtain
\begin{align*}
\Theta(a)
&=
\frac{1}{\LnL(\alpha,a)}
\frac{\partial}{\partial\alpha}\LnL(\alpha,a)
\nonumber\\
&=
\frac{1}{\LnL(\alpha,a)}
\int_0^\infty
n[\Phic(\sigma)]^{n-1}
\frac{\partial}{\partial\alpha}\Phic(\sigma)
e^{-na\sigma}\,d\sigma
\nonumber\\
&=
\int_0^\infty
n\,
\frac{\frac{\partial}{\partial\alpha}\Phic(\sigma)}
{\Phic(\sigma)}
\frac{[\Phic(\sigma)]^n e^{-na\sigma}}
{\LnL(\alpha,a)}
\,d\sigma
\nonumber\\
&=
\mathbb E_{\mathcal Q_{\alpha,a}^{(L)}}
\left[
n\,\frac{\partial}{\partial\alpha}\log\Phic(\sigma)
\right],
\end{align*}
where \(\mathcal{Q}^{(L)}_{\alpha,a}\) is the probability measure on
\((0,\infty)\) defined by
\[
d\mathcal{Q}^{(L)}_{\alpha,a}(\sigma)
=
\frac{[\Phic(\sigma)]^n e^{-na\sigma}}
{\LnL(\alpha,a)}\,d\sigma.
\]

By Lemma~\ref{lem:cross}, the map
\[
\sigma\longmapsto
n\,\frac{\partial}{\partial\alpha}\log\Phic(\sigma)
\]
is strictly decreasing.

Now fix \(a'>a\). We have
\[
\frac{d\mathcal Q^{(L)}_{\alpha,a'}}
{d\mathcal Q^{(L)}_{\alpha,a}}(\sigma)
=
\frac{\LnL(\alpha,a)}{\LnL(\alpha,a')}
e^{-n(a'-a)\sigma},
\]
which is also strictly decreasing in \(\sigma\) and has expectation \(1\)
under \(\mathcal Q^{(L)}_{\alpha,a}\). Using Chebyshev's covariance inequality, we obtain
\begin{align*}
&\quad\Theta(a')-\Theta(a)\\
&=
\mathbb E_{\mathcal Q^{(L)}_{\alpha,a}}
\left[
n\,\frac{\partial}{\partial\alpha}\log\Phic(\sigma)
\frac{d\mathcal Q^{(L)}_{\alpha,a'}}
{d\mathcal Q^{(L)}_{\alpha,a}}(\sigma)
\right]\\
&\qquad\qquad\qquad-
\mathbb E_{\mathcal Q^{(L)}_{\alpha,a}}
\left[
n\,\frac{\partial}{\partial\alpha}\log\Phic(\sigma)
\right] \\
&=
\mathbb E_{\mathcal Q^{(L)}_{\alpha,a}}
\left[
n\,\frac{\partial}{\partial\alpha}\log\Phic(\sigma)
\frac{d\mathcal Q^{(L)}_{\alpha,a'}}
{d\mathcal Q^{(L)}_{\alpha,a}}(\sigma)
\right] \\
&\quad-
\mathbb E_{\mathcal Q^{(L)}_{\alpha,a}}
\left[
n\,\frac{\partial}{\partial\alpha}\log\Phic(\sigma)
\right]
\mathbb E_{\mathcal Q^{(L)}_{\alpha,a}}
\left[
\frac{d\mathcal Q^{(L)}_{\alpha,a'}}
{d\mathcal Q^{(L)}_{\alpha,a}}(\sigma)
\right] \\
&=
\operatorname{Cov}_{\mathcal Q^{(L)}_{\alpha,a}}
\left(
n\,\frac{\partial}{\partial\alpha}\log\Phic(\sigma),
\frac{d\mathcal Q^{(L)}_{\alpha,a'}}
{d\mathcal Q^{(L)}_{\alpha,a}}(\sigma)
\right)>0.
\end{align*}
The inequality is strict because both functions are strictly decreasing and
\(\mathcal Q^{(L)}_{\alpha,a}\) has a positive density on
\((0,\infty)\). Hence, \(a\mapsto\Theta(a)\) is strictly increasing.

Finally, Eq.~\eqref{eq_200} gives
\[
\frac{\partial}{\partial\alpha}\log\Phic(\sigma)
=
-\E_{\PP^{(L)}_{\sigma,\alpha}}[\mathcal{D}_\sigma]
<0.
\]
This shows that
\[
\Theta(a)
=
\mathbb E_{\mathcal Q_{\alpha,a}^{(L)}}
\left[
n\,\frac{\partial}{\partial\alpha}\log\Phic(\sigma)
\right]
<0,
\]
for every \(a>0\).
\end{proof}

Now, we are ready for our first main result.
\begin{theorem}\label{thm:alpha}
The map $\alpha\mapsto\bc(\alpha)$ is strictly decreasing on $(0,\infty)$. 
\end{theorem}

\begin{proof}
Rewrite \eqref{eq:betaexplicit} as
\begin{equation}\label{main:form}
\begin{split}
&\bc(\alpha)=\\
&\frac{(1+2\mu)^2+\mu(3+2\mu)\mathcal R_3(n,\alpha)-(1+\mu)(1+2\mu)\mathcal R_2(n,\alpha)}{(1+\mu)(1+2\mu)\mathcal R_2(n,\alpha)+\mu(3+2\mu)\mathcal R_3(n,\alpha)-(1+2\mu)},
\end{split}
\end{equation}
where
\[
\mathcal{R}_i(n,\alpha)=\frac{\LnL(\alpha,a_i)}{\LnL(\alpha,a_1)},
\]
for $i=2,3$. Here, $a_1=1/(2\nu)$, $a_2=(1+2\mu)/(2\nu)$ and $a_3=(3+2\mu)/(2\nu)$. By Lemma \ref{lem:X_monotone}, note that
\[
\frac{\partial}{\partial\alpha}\log \mathcal{R}_i(n,\alpha)=\Theta(a_i)-\Theta(a_1)>0,
\]
as $a_i>a_1$, for $i=2,3$. This shows that $\alpha\to \mathcal{R}_i(n,\alpha)$ is increasing.

Now, differentiating $\bc(\alpha)$ with respect to $\mathcal{R}_2$ leads
\[
\begin{aligned}
&\frac{\partial \bc(\alpha)}{\partial \mathcal{R}_2}=\\
&-\frac{(1+\mu)(1+2\mu)\bigl(2\mu(1+2\mu)+2\mu(3+2\mu)\mathcal{R}_3\bigr)}
{\bigl((1+\mu)(1+2\mu)\mathcal{R}_2+\mu(3+2\mu)\mathcal{R}_3-(1+2\mu)\bigr)^2}<0.
\end{aligned}
\]
In addition, we have
\[
\begin{aligned}
&\frac{\partial \bc(\alpha)}{\partial \mathcal{R}_3}=\\
&\frac{2\mu(3+2\mu)(1+\mu)(1+2\mu)(\mathcal{R}_2-1)}{\bigl((1+\mu)(1+2\mu)\mathcal{R}_2+\mu(3+2\mu)\mathcal{R}_3-(1+2\mu)\bigr)^2}.
\end{aligned}
\]
On the other hand, since  $x\mapsto\LnL(\alpha,x)$ is strictly decreasing and $a_2>a_1$, we deduce that 
\[
\mathcal{R}_2(n,\alpha)=\frac{\LnL(\alpha,a_2)}{\LnL(\alpha,a_1)}<1,
\]
and therefore
\[
\frac{\partial \bc}{\partial \mathcal{R}_3}<0.
\]
So $\bc(\alpha)$ decreases in each of the two variables $\mathcal{R}_2$ and $\mathcal{R}_3$, while both variables increase with respect to $\alpha$. It follows that
\[
\frac{\partial}{\partial\alpha}\bc(\alpha)<0.
\]
\end{proof}

\begin{corollary}\label{cor:alpha}
The threshold reaches its minimum in the limit $\alpha\to\infty$. In addition, we have
\[
\lim_{\alpha\to0^+}\bc(\alpha)=\infty.
\]
\end{corollary}

\begin{proof}
By the previous theorem, it is clear that the threshold reaches its minimum in the limit $\alpha\to\infty$. Now, consider the case $\alpha\to0^{+}$. By the bounded convergence theorem, we have
\[
\lim_{\alpha\to0^+}\Phic(\sigma)=\lim_{\alpha\to0^+}\E_0\left[e^{-\alpha \mathcal{D}_\sigma}\right]= 1,
\]
for every $\sigma\ge0$. This shows that 
\[
\lim_{\alpha\to0^+}\left(\Phic(\sigma)\right)^{n}e^{-nx\sigma}=e^{-nx\sigma}.
\]
In addition, we have
\[
0\leq [\Phic(\sigma)]^n e^{-nx\sigma}\leq e^{-nx\sigma},
\]
where $x\mapsto e^{-nx\sigma}$ is integrable on $(0,\infty)$.
Then the Lebesgue dominated convergence theorem yields
\[
\lim_{\alpha\to0^+}\LnL(\alpha,x)=\int_0^\infty e^{-nx\sigma}\,d\sigma=\frac1{nx}.
\]
Evaluating this limit for $x=1/(2\nu),(1+2\mu)/(2\nu),(3+2\mu)/(2\nu)$ and substituting the results into Eq.~\eqref{eq:betaexplicit}, the denominator of $\bc(\alpha)$ converges to
\[
\begin{aligned}
&-(1+2\mu)\frac{2\nu}{n}
+(1+\mu)(1+2\mu)\frac{2\nu}{n(1+2\mu)}\\
&\quad+\mu(3+2\mu)\frac{2\nu}{n(3+2\mu)}
=0
\end{aligned}
\]
and the numerator converges to
\[
\begin{aligned}
&(1+2\mu)^2\frac{2\nu}{n}
+\mu(3+2\mu)\frac{2\nu}{n(3+2\mu)}\\
&
\quad-(1+\mu)(1+2\mu)\frac{2\nu}{n(1+2\mu)}
=
\frac{8\mu(1+\mu)\nu}{n}>0.
\end{aligned}
\]
Therefore, we deduce that
\begin{equation}\label{eq:alpha_zero_threshold_limit}
\lim_{\alpha\to0^+}\bc(\alpha)=\infty.
\end{equation}

\end{proof}


\section{\label{AppendixC}Effect of Phenotype-Space Dimension}
In this section, we will study how the dimension of the phenotype space affects the threshold $\beta_n^{(L)}(\alpha)$.
The remainder is organized as follows. Lemma~\ref{lem:variance} first establishes a strict convexity property for the logarithm of the function appearing in the integral representation of the threshold. Lemma~\ref{lem:g_negative} then uses this convexity to determine the sign and monotonicity of an auxiliary function required in the comparison argument. Proposition~\ref{prop:dimension} applies a covariance inequality to show that the relevant ratios of integrals are strictly increasing with the phenotype-space dimension. Theorem~\ref{thm:dimension} combines this monotonicity with the explicit threshold formula to prove that the threshold is strictly decreasing with the dimension. Finally, Corollary~\ref{cor:n_infty} derives the large-dimension asymptotic behavior of the underlying integrals and determines the limiting value of the threshold as the dimension tends to infinity.

\begin{lemma}\label{lem:variance}
For every $\alpha>0$ and every $\sigma>0$, we have
\[
\frac{\partial^2}{\partial \sigma^2}\log \Phic(\sigma)>0.
\]
\end{lemma}

\begin{proof}
Differentiating Eq.~\eqref{eq:Phidef} yields
\begin{equation}\label{eq27}
\begin{split}
\frac{\partial}{\partial \sigma}\log \Phic(\sigma)&=\frac{\frac{\partial}{\partial \sigma} \Phic(\sigma)}{\Phic(\sigma)}\\
&=-\sum_{p=0}^{L-1}(1-c_p)\frac{\hatphi(p)e^{-\sigma(1-c_p)}}{L\,\Phic(\sigma)}\\
&=-\E_{\tilde{\PP}^{(L)}_{\alpha,\sigma}}[1-c_p],
\end{split}
\end{equation}
where $\tilde{\PP}^{(L)}_{\alpha,\sigma}$ is the probability measure on $\{0,1,\ldots,L-1\}$ given by 
\[
\tilde{\PP}^{(L)}_{\alpha,\sigma}(p)
=
\frac{\hatphi(p)e^{-\sigma(1-c_p)}}{L\,\Phic(\sigma)}
\]
and $c_p=\cos\left(2\pi p/L\right)$.
Differentiating once more with respect to $\sigma$, we obtain
\begin{equation}\label{eq20}
\begin{split}
&\frac{\partial^2}{\partial \sigma^2}\log \Phic(\sigma)\\
=&\sum_{p=0}^{L-1}(1-c_p)^2\frac{\hatphi(p)e^{-\sigma(1-c_p)}}{L\,\Phic(\sigma)}\\
&+\frac{\frac{\partial}{\partial \sigma} \Phic(\sigma)}{\Phic(\sigma)}\sum_{p=0}^{L-1}(1-c_p)\frac{\hatphi(p)e^{-\sigma(1-c_p)}}{L\,\Phic(\sigma)}\\
=&\E_{\tilde{\PP}^{(L)}_{\alpha,\sigma}}\left[(1-c_p)^2\right]-\left(\E_{\tilde{\PP}^{(L)}_{\alpha,\sigma}}[1-c_p]\right)^2\\
=&\Var_{\tilde{\PP}^{(L)}_{\alpha,\sigma}}(1-c_p)=\Var_{\tilde{\PP}^{(L)}_{\alpha,\sigma}}(c_p).
\end{split}
\end{equation}
On the other hand, using Lemma \ref{lem:phi_properties}, all coefficients $\{\hatphi(p)\}_{p=0}^{L-1}$ are strictly positive. Then, the measure $\tilde{\PP}^{(L)}_{\alpha,\sigma}$ gives positive weight to at least two different values of $c_p$, so the variance in \eqref{eq20} is strictly positive.
\end{proof}

The next lemma introduces the function that controls the dependence on $n$.

\begin{lemma}\label{lem:g_negative}
Define
\begin{equation}\label{eq21}
g_\alpha^{(L)}(\sigma)=\log \Phic(\sigma)-\sigma \frac{\partial}{\partial \sigma}\log \Phic(\sigma).
\end{equation}
for any $\sigma>0$.
Then, we have $g_\alpha^{(L)}(\sigma)<0$, for every $\sigma>0$. Moreover, $\sigma \mapsto g_\alpha^{(L)}(\sigma)$ is strictly decreasing on $(0,\infty)$.
\end{lemma}

\begin{proof}
Differentiating Eq.~\eqref{eq21} yields
\[
\frac{\partial}{\partial\sigma}g_\alpha^{(L)}(\sigma)
=
-\sigma \frac{\partial^2}{\partial \sigma^2}\log \Phic(\sigma).
\]
Using Lemma \ref{lem:variance}, we deduce that \(\frac{\partial}{\partial\sigma}g_\alpha^{(L)}(\sigma)<0\) for every \(\sigma>0\), so $\sigma \mapsto g_\alpha^{(L)}(\sigma)$ is strictly decreasing.

To show that $g_\alpha^{(L)}(\sigma)<0$, note the continuity of $\sigma \mapsto g_\alpha^{(L)}(\sigma)$ on $[0,\infty)$. In addition, since $\sigma \mapsto g_\alpha^{(L)}(\sigma)$ decreases strictly, we conclude that 
\[
g_\alpha^{(L)}(\sigma)<g_\alpha^{(L)}(0)=0
\qquad\text{for every }\sigma>0.
\]
\end{proof}

\begin{proposition}\label{prop:dimension}
Fix \(b>a>0\). Then
\[
n\mapsto \Lambda_{n,\alpha}^{(L)}(a,b):=\frac{\LnL(\alpha,b)}{\LnL(\alpha,a)}
\]
is strictly increasing on $(0,\infty)$.
\end{proposition}

\begin{proof}
By a simple substitution, $\LnL$ can be written as
\[
\LnL(\alpha,x)=\int_0^\infty [\Phic(\sigma)]^n e^{-nx\sigma}\,d\sigma
=
\frac{1}{n}f_{\alpha}^{(L)}(n,x),
\]
where \[f_{\alpha}^{(L)}(n,x)=\int_0^\infty \bigl[\Phic(u/n)\bigr]^ne^{-xu}\,du.\]
Then, we deduce that
\[
\Lambda_{n,\alpha}^{(L)}(a,b)=\frac{f_{\alpha}^{(L)}(n,b)}{f_{\alpha}^{(L)}(n,a)}.
\]
Differentiating this expression with respect to $n$, we have
\begin{align*}
&\frac{\partial}{\partial n} f_{\alpha}^{(L)}(n,x)\\
=&\frac{\partial}{\partial n}\int_0^\infty e^{n\log\Phic(u/n)}e^{-xu}\,du\\
=&\int_0^\infty\left(\log\Phic\left(\frac{u}{n}\right)-\frac{u}{n}\frac{\partial}{\partial\sigma}\log\Phic\left(\frac{u}{n}\right)\right)\\
&\qquad\qquad\qquad\quad\times\bigl[\Phic(u/n)\bigr]^ne^{-xu}\,du\\
=&\int_0^\infty g_\alpha^{(L)}\left(u/n\right)\bigl[\Phic(u/n)\bigr]^ne^{-xu}\,du.
\end{align*}
Then, we deduce that
\begin{align}
\frac{\partial}{\partial n}\log f_{\alpha}^{(L)}(n,x)
=&\int_0^\infty g_\alpha^{(L)}\left(u/n\right)\frac{\bigl[\Phic(u/n)\bigr]^ne^{-xu}}{f_{\alpha}^{(L)}(n,x)}\,du\nonumber\\
=&\mathbb{E}_{\tilde{\mathcal{Q}}^{(L)}_{n,x}}\left[g_\alpha^{(L)}\left(u/n\right)\right],\label{eq22}
\end{align}
where \(\tilde{\mathcal{Q}}^{(L)}_{n,x}\) is the probability measure on \((0,\infty)\) defined by
\[
d\tilde{\mathcal{Q}}^{(L)}_{n,x}(u)
=
\frac{\bigl[\Phic(u/n)\bigr]^ne^{-xu}}{f_{\alpha}^{(L)}(n,x)}\,du.
\]
Differentiating once more \eqref{eq22} with respect to $x$, we get
\begin{equation}\label{eq23}
\begin{split}
&\frac{\partial^2}{\partial x\partial n}\log f_{\alpha}^{(L)}(n,x)\\
&=-\int_0^\infty ug_\alpha^{(L)}\left(\frac{u}{n}\right)\frac{\bigl[\Phic\left(\frac{u}{n}\right)\bigr]^ne^{-xu}}{f_{\alpha}^{(L)}(n,x)}\,du\\
&\quad
-\int_0^\infty g_\alpha^{(L)}\left(\frac{u}{n}\right)\frac{\bigl[\Phic\left(\frac{u}{n}\right)\bigr]^ne^{-xu}}{f_{\alpha}^{(L)}(n,x)}\,du\frac{\frac{\partial}{\partial x}f_{\alpha}^{(L)}(n,x)}{f_{\alpha}^{(L)}(n,x)}\\
&=-\mathbb{E}_{\tilde{\mathcal{Q}}^{(L)}_{n,x}}\left[ug_\alpha^{(L)}\left(\frac{u}{n}\right)\right]-\mathbb{E}_{\tilde{\mathcal{Q}}^{(L)}_{n,x}}\left[g_\alpha^{(L)}\left(\frac{u}{n}\right)\right]\frac{\frac{\partial}{\partial x}f_{\alpha}^{(L)}(n,x)}{f_{\alpha}^{(L)}(n,x)}.
\end{split}
\end{equation}
On the other hand, note
\begin{equation*}
\begin{split}
\frac{\frac{\partial}{\partial x}f_{\alpha}^{(L)}(n,x)}{f_{\alpha}^{(L)}(n,x)}&=-\int_0^\infty u\frac{\bigl[\Phic(u/n)\bigr]^ne^{-xu}}{f_{\alpha}^{(L)}(n,x)}\,du\\
&=-\mathbb{E}_{\tilde{\mathcal{Q}}^{(L)}_{n,x}}\left[u\right].
\end{split}
\end{equation*}
Inserting this identity in Eq.~\eqref{eq23} leads to
\begin{equation*}
\frac{\partial^2}{\partial x\partial n}\log f_{\alpha}^{(L)}(n,x)=-\Cov_{\tilde{\mathcal{Q}}^{(L)}_{n,x}}\left(u,\,g_\alpha^{(L)}(u/n)\right).
\end{equation*}
Since \(u\mapsto u\) is increasing and \(u\mapsto g_\alpha^{(L)}(u/n)\) is strictly decreasing (Lemma \ref{lem:g_negative}), Chebyshev's covariance inequality shows that
\begin{equation}\label{eq24}
\frac{\partial^2}{\partial x\partial n}\log f_{\alpha}^{(L)}(n,x)=-\Cov_{\tilde{\mathcal{Q}}^{(L)}_{n,x}}\left(u,\,g_\alpha^{(L)}(u/n)\right)>0.
\end{equation}
The inequality is strict since neither function is \(\tilde{\mathcal{Q}}^{(L)}_{n,x}\)-a.s.\ constant.
Integrating \eqref{eq24} with respect to $x$ from $a$ to $b$, we find
\[
\begin{aligned}
\frac{\partial}{\partial n}\log \Lambda_{n,\alpha}^{(L)}(a,b)&=\frac{\partial}{\partial n}\log\frac{f_{\alpha}^{(L)}(n,b)}{f_{\alpha}^{(L)}(n,a)}\\
&=\int_a^b\frac{\partial^2}{\partial x\,\partial n}\log f_{\alpha}^{(L)}(n,x)\,dx>0,
\end{aligned}
\]
which completes the proof.
\end{proof}


\begin{theorem}\label{thm:dimension}
For every integer $n\ge1$, we have
\[
\beta_{n+1}^{(L)}(\alpha)<\beta_n^{(L)}(\alpha).
\]
In fact, the map $n\longmapsto \beta_n^{(L)}(\alpha)$ extends to a strictly decreasing function on $(0,\infty)$.
\end{theorem}

\begin{proof}
Differentiating \eqref{main:form} with respect to $n$ leads to
\begin{equation}
\frac{\partial}{\partial n}\bc(\alpha)=\frac{\partial \bc}{\partial \mathcal{R}_2}\times\frac{\partial}{\partial n} \mathcal{R}_2(n,\alpha)+\frac{\partial \bc}{\partial \mathcal{R}_3}\times\frac{\partial}{\partial n} \mathcal{R}_3(n,\alpha).
\end{equation}
In the proof of Theorem \ref{thm:alpha}, we have shown that \(\dfrac{\partial \bc}{\partial \mathcal{R}_2}<0\) and \(\dfrac{\partial \bc}{\partial \mathcal{R}_3}<0\).
To complete the proof, observe that $\mathcal{R}_2(n,\alpha)=\Lambda_{n,\alpha}^{(L)}(a_1,a_2)$ and $\mathcal{R}_3(n,\alpha)=\Lambda_{n,\alpha}^{(L)}(a_1,a_3)$, which are increasing with respect to $n$ by Proposition \ref{prop:dimension}.
\end{proof}

The next theorem gives the limiting threshold when the dimension becomes very large.

\begin{theorem}\label{cor:n_infty}
We have
\begin{equation}\label{eq:beta_infty}
\lim_{n\to\infty}\beta_n^{(L)}(\alpha)=\frac{4\gamma_\alpha^2\nu^2+(8\mu+6)\gamma_\alpha\nu+4\mu^2+8\mu+3}{4\gamma_\alpha\nu(\gamma_\alpha\nu+\mu+1)}.
\end{equation}
\end{theorem}

\begin{proof}
To establish this limit, we first determine the asymptotic behavior of $\mathcal{L}_n(\alpha,x)$ for $n\to\infty$.
More precisely, we will show 
\[
n\LnL(\alpha,x) \sim\frac{1}{x+\gamma_\alpha}, 
\]
as $n\to\infty$. Recall the integral representation 
\begin{equation}\label{eq25}
n\LnL(\alpha,x)=\int_0^\infty \exp\Bigl(n\log \Phic(t/n)-xt\Bigr)\,dt.
\end{equation}
Because $\Phic$ is constructed as a finite sum of exponentials, it is $C^\infty$ on $[0,\infty)$. By Lemma \ref{lem:phi_properties}(iv), its local behavior near zero is governed by the expansion
\begin{equation}\label{eq26}
\log \Phic(\sigma) = -\gamma_\alpha \sigma + O(\sigma^2),
\end{equation}
as $\sigma \downarrow 0$.
Fix a small constant $\varepsilon \in (0,x+\gamma_\alpha)$. There exists $\delta>0$ such that
\begin{equation}\label{eq:local_log_bound}
-(\gamma_\alpha+\varepsilon)\sigma \le \log \Phic(\sigma) \le -(\gamma_\alpha-\varepsilon)\sigma, 
\end{equation}
for $0\le \sigma\le \delta$.
Splitting the integral in \eqref{eq25} into the sum of a local component, $I_n$, and a tail component, $J_n$,
where
\begin{align*}
I_n&=\int_0^{n\delta}\exp\!\Bigl(n\log \Phic(t/n)-xt\Bigr)\,dt,\\
J_n&=\int_{n\delta}^\infty\exp\!\Bigl(n\log \Phic(t/n)-xt\Bigr)\,dt.
\end{align*}

Let us first examine the main contribution, $I_n$. For any fixed $t\ge 0$, a consequence of \eqref{eq26} is
\[
n\log \Phic(t/n)\sim -\gamma_\alpha t,
\]
as $n\to\infty$.
Therefore, we have
\[
\lim_{n\to\infty}\exp\!\Bigl(n\log \Phic(t/n)-xt\Bigr)=e^{-(x+\gamma_\alpha)t}.
\]
Moreover, for \(0\le t\le n\delta\), Eq.~\eqref{eq:local_log_bound} leads to
\[
0\le \exp\!\Bigl(n\log \Phic(t/n)-xt\Bigr)\le e^{-(x+\gamma_\alpha-\varepsilon)t}.
\]
As the right-hand side is integrable on \([0,\infty)\), the Lebesgue dominated convergence theorem yields
\[
\lim_{n\rightarrow\infty}I_n=\int_0^\infty e^{-(x+\gamma_\alpha)t}\,dt=\frac{1}{x+\gamma_\alpha}.
\]

It remains to show that the tail term \(J_n\) is negligible as
\(n\to\infty\). Since \(\sigma\mapsto \Phic(\sigma)\) is decreasing on
\([0,\infty)\) (Lemma \ref{lem:phi_properties}(\textit{v})), we have, for every \(t\ge n\delta\),
\[
\Phic(t/n)\le \Phic(\delta)=:r.
\]
Moreover, \(0<r<1\), because \(\delta>0\) and \(\Phic(0)=1\). Hence
\[
0\le
\exp\!\Bigl(n\log \Phic(t/n)-xt\Bigr)
=
\Phic(t/n)^n e^{-xt}
\le r^n e^{-xt}.
\]
Therefore,
\[
0\le J_n
\le r^n\int_{n\delta}^{\infty}e^{-xt}\,dt
=
\frac{r^n}{x}e^{-xn\delta}.
\]
Since \(0<r<1\) and \(x,\delta>0\), the right-hand side tends to zero as
\(n\to\infty\). By the squeeze theorem,
\[
\lim_{n\to\infty}J_n=0.
\]

By combining the dominant term $I_n$ with the vanishing tail $J_n$, we have the asymptotic expansion
\[
\LnL(\alpha,x)=\frac{1}{n(x+\gamma_\alpha)}+o\!\left(\frac1n\right)
\qquad\text{as }n\to\infty.
\]
Evaluating this asymptotic formula for $x=\frac{1}{2\nu},\,\frac{1+2\mu}{2\nu},\,\frac{3+2\mu}{2\nu}$,
and substituting the results into the explicit formula \eqref{eq:betaexplicit}, we obtain
\[
\lim_{n\to\infty}\beta_n^{(L)}(\alpha)
=
\frac{4\gamma_\alpha^2\nu^2+(8\mu+6)\gamma_\alpha\nu+4\mu^2+8\mu+3}
{4\gamma_\alpha\nu(\gamma_\alpha\nu+\mu+1)}.
\]
This completes the proof.
\end{proof}


\section{\label{AppendixD}Effect of the Phenotype Mutation Rate}
In this section, we will study the effect of the phenotype mutation rate $\nu$ on the threshold $\beta_n^{(L)}(\alpha)$. To keep the notation light, we continue to write \(\beta_n^{(L)}(\alpha)\), although this quantity also depends on the fixed strategy mutation rate \(\mu\) and on the phenotype mutation rate \(\nu\).

\begin{lemma}\label{Lemma_D1}
We have the following expansion
\begin{equation}\label{eq28}
\LnL(\alpha,x)=\frac{1}{nx}-\frac{\gamma_\alpha}{nx^2}+\mathcal{O}(x^{-3}),
\end{equation}
as $x\to\infty$.
\end{lemma}
\begin{proof}
By Lemma~\ref{lem:phi_properties}, we have $\Phi_\alpha^{(L)}(0)=1$ and $\frac{\partial}{\partial\sigma} \Phi_\alpha^{(L)}(0)=-\gamma_\alpha$.
In addition, since \(\Phi_\alpha^{(L)}\) is a finite sum of exponentials, it is \(C^2\) near
\(0\). Hence, for some \(\delta>0\), we have
\[
\Phi_\alpha^{(L)}(\sigma)=1-\gamma_\alpha\sigma+\mathcal{O}(\sigma^2),
\]
for $0<\sigma<\delta$. This means that
\[
\left[\Phi_\alpha^{(L)}(\sigma)\right]^n=1-n\gamma_\alpha\sigma+\mathcal{O}(\sigma^2).
\]

Now split the integral representation as
\begin{align}
&\LnL(\alpha,x)\label{eqq1}\\
&=\int_0^\delta \left[\Phi_\alpha^{(L)}(\sigma)\right]^n e^{-nx\sigma}\,d\sigma+\int_\delta^\infty \left[\Phi_\alpha^{(L)}(\sigma)\right]^n e^{-nx\sigma}\,d\sigma.\nonumber
\end{align}
On \([0,\delta]\), the previous expansion gives
\[
\int_0^\delta \left[\Phi_\alpha^{(L)}(\sigma)\right]^n e^{-nx\sigma}\,d\sigma=\int_0^\delta\bigl(1-n\gamma_\alpha\sigma+\mathcal{O}(\sigma^2)\bigr)e^{-nx\sigma}\,d\sigma .
\]
On the other hand, note that
\[
\begin{aligned}
\int_0^\infty e^{-nx\sigma}\,d\sigma&=\frac1{nx},\\
\int_0^\infty \sigma e^{-nx\sigma}\,d\sigma&=\frac1{(nx)^2},\\
\int_0^\infty \sigma^2 e^{-nx\sigma}\,d\sigma&=\mathcal{O}(x^{-3}),
\end{aligned}
\]
and that replacing \([0,\delta]\) by \([0,\infty)\) only creates an exponentially small error. Then, we obtain
\begin{equation}\label{eqq2}
\int_0^\delta \left[\Phi_\alpha^{(L)}(\sigma)\right]^n e^{-nx\sigma}\,d\sigma =\frac1{nx} - \frac{\gamma_\alpha}{n x^2}+\mathcal{O}(x^{-3}).
\end{equation}

On \([\delta,\infty)\), since \(0\le \left[\Phi_\alpha^{(L)}(\sigma)\right]^n\le1\), we have
\begin{equation}\label{eqq3}
\begin{split}
0\le&\int_\delta^\infty \left[\Phi_\alpha^{(L)}(\sigma)\right]^n e^{-nx\sigma}\,d\sigma\\
&\qquad\le\int_\delta^\infty e^{-nx\sigma}\,d\sigma=\frac{e^{-nx\delta}}{nx}=
\mathcal{O}(x^{-3}).
\end{split}
\end{equation}
Combining the local and tail estimates in \eqref{eqq2} and \eqref{eqq3} in \eqref{eqq1} yields
\[
\LnL(\alpha,x)
=
\frac{1}{nx}
-
\frac{\gamma_\alpha}{n x^2}
+
\mathcal{O}(x^{-3}),
\]
as $x\to\infty$.
\end{proof}

\begin{lemma}\label{lem:Laplace_small_x}
As \(x\to0^+\), we have
\begin{equation}
\LnL(\alpha,x)=\frac{\left(\hatphi(0)\right)^n}{nL^n x}+C_n^{(L)}(\alpha)+\mathcal{O}(x),
\end{equation}
where
\begin{equation}
C_n^{(L)}(\alpha)=\frac{1}{L^n}\sum_{\substack{(p_1,\ldots,p_n)\in\{0,\ldots,L-1\}^n\\(p_1,\ldots,p_n)\ne(0,\ldots,0)}}\frac{\prod_{j=1}^n \hatphi(p_j)}{n-\sum_{j=1}^n c_{p_j}}.
\end{equation}
\end{lemma}

\begin{proof}
From the spectral representation,
\[
\LnL(\alpha,x)=\frac{1}{L^n}\sum_{p_1,\ldots,p_n=0}^{L-1}\frac{\prod_{j=1}^n \hatphi(p_j)}{n-\sum_{j=1}^n c_{p_j}+nx}.
\]
Since \(c_0=1\), the term corresponds to the zero mode \((p_1,\ldots,p_n)=(0,\ldots,0)\) equals
\[
\frac{1}{L^n}\frac{\left(\hatphi(0)\right)^n}{nx}=\frac{\left(\hatphi(0)\right)^n}{nL^n x}.
\]
For every nonzero mode, we have
\(
n-\sum_{j=1}^n c_{p_j}>0.
\)
Hence, as \(x\to0^+\), we obtain
\[
\frac{1}{n-\sum_{j=1}^n c_{p_j}+nx}=\frac{1}{n-\sum_{j=1}^n c_{p_j}}+\mathcal{O}(x).
\]
This proves the desired expansion.
\end{proof}

\begin{theorem}\label{thm:nu_asymptotics}
Fix $\alpha>0$, $L\ge2$, $n\ge1$, and $\mu>0$. Then, as $\nu\to0^+$, we have
\begin{equation}\label{eq:beta_nu_small}
\beta_n^{(L)}(\alpha)\sim\frac{(1+2\mu)(3+2\mu)}{4(1+\mu)\gamma_\alpha}\,\frac1{\nu}.
\end{equation}
In addition, as $\nu\to\infty$, we have
\begin{equation}\label{eq:beta_nu_large}
\beta_n^{(L)}(\alpha)\sim\frac{2\,\left(\hatphi(0)\right)^n}{nL^n\,C_n^{(L)}(\alpha)}\nu.
\end{equation}
\end{theorem}

\begin{proof}
First, let $\nu\to0^+$. Evaluating the expansion \eqref{eq28} in Lemma \ref{Lemma_D1} for $x_1=\frac{1}{2\nu}$, $x_2=\frac{1+2\mu}{2\nu}$, and $x_3=\frac{3+2\mu}{2\nu}$ and substituting these expansions into \eqref{eq:betaexplicit}, the numerator of $\beta_n^{(L)}(\alpha)$ satisfies
\begin{align*}
&(1+2\mu)^2\LnL(\alpha,x_1)+\mu(3+2\mu)\LnL(\alpha,x_3)\\
&-(1+\mu)(1+2\mu)\LnL(\alpha,x_2)\\
=&(1+2\mu)^2\frac{2\nu}{n}+\mu(3+2\mu)\frac{2\nu}{n(3+2\mu)}\\
&-(1+\mu)(1+2\mu)\frac{2\nu}{n(1+2\mu)}+\mathcal{O}(\nu^2)\\
=&\frac{8\mu(1+\mu)}{n}\nu+\mathcal{O}(\nu^2),
\end{align*}
while the denominator can be expanded as 
\begin{align*}
&(1+\mu)(1+2\mu)\LnL(\alpha,x_2)+\mu(3+2\mu)\LnL(\alpha,x_3)\\
&-(1+2\mu)\LnL(\alpha,x_1)\\
=&(1+\mu)(1+2\mu)\left[\frac{2\nu}{n(1+2\mu)}-\frac{4\gamma_\alpha\nu^2}{n(1+2\mu)^2}\right]\\
&+\mu(3+2\mu)\left[\frac{2\nu}{n(3+2\mu)}-\frac{4\gamma_\alpha\nu^2}{n(3+2\mu)^2}\right]\\
&-(1+2\mu)\left[\frac{2\nu}{n}-\frac{4\gamma_\alpha\nu^2}{n}\right]+\mathcal{O}(\nu^3)\\
&=\frac{32\gamma_\alpha\mu(1+\mu)^2}{n(1+2\mu)(3+2\mu)}\nu^2+\mathcal{O}(\nu^3).
\end{align*}
Dividing these expansions yields
\begin{equation}
\beta_n^{(L)}(\alpha)=\frac{(1+2\mu)(3+2\mu)}{4(1+\mu)\gamma_\alpha}\frac{1}{\nu}+\mathcal{O}(1).
\end{equation}

Now, let $\nu\to\infty$. Using Lemma \ref{lem:Laplace_small_x} for \(x=\frac{1}{2\nu},\frac{1+2\mu}{2\nu},\frac{3+2\mu}{2\nu}\), and substituting the results into \eqref{eq:betaexplicit}, the numerator of \(\beta_n^{(L)}(\alpha)\) will take the form
\begin{align*}
&(1+2\mu)^2\LnL(\alpha,x_1)+\mu(3+2\mu)\LnL(\alpha,x_3)\\
&-(1+\mu)(1+2\mu)\LnL(\alpha,x_2)\\
=&(1+2\mu)^2\frac{2\left(\hatphi(0)\right)^n\nu}{nL^n}+\mu(3+2\mu)\frac{2\left(\hatphi(0)\right)^n\nu}{nL^n(3+2\mu)}\\
&-(1+\mu)(1+2\mu)  \frac{2\left(\hatphi(0)\right)^n\nu}{nL^n(1+2\mu)}+\mathcal{O}(1)\\
=&\frac{8\left(\hatphi(0)\right)^n\mu(1+\mu)}{nL^n}\nu+\mathcal{O}(1).
\end{align*}

For the denominator, the leading \(O(\nu)\) terms cancel and we obtain
\begin{align*}
&(1+\mu)(1+2\mu)\LnL(\alpha,x_2)+\mu(3+2\mu)\LnL(\alpha,x_3)\\
&\qquad-(1+2\mu)\LnL(\alpha,x_1)\\
=&(1+\mu)(1+2\mu)\left[\frac{2\left(\hatphi(0)\right)^n\nu}{nL^n(1+2\mu)}+C_n^{(L)}(\alpha)\right]\\
&+\mu(3+2\mu)\left[\frac{2\left(\hatphi(0)\right)^n\nu}{nL^n(3+2\mu)}+C_n^{(L)}(\alpha)\right]\\
&-(1+2\mu)\left[\frac{2\left(\hatphi(0)\right)^n\nu}{nL^n}+C_n^{(L)}(\alpha)\right]+\mathcal{O}(\nu^{-1})\\
=&4\mu(1+\mu)C_n^{(L)}(\alpha)+\mathcal{O}(\nu^{-1}).
\end{align*}

Dividing these expansions gives
\begin{equation}
\beta_n^{(L)}(\alpha)=\frac{2\left(\hatphi(0)\right)^n}{nL^nC_n^{(L)}(\alpha)}\nu+\mathcal{O}(1).
\end{equation}
\end{proof}

\begin{remark}
The asymptotic expansions \eqref{eq:beta_nu_small} and
\eqref{eq:beta_nu_large} show that
\[
\lim_{\nu\to0^+}\beta_n^{(L)}(\alpha)
=
\lim_{\nu\to\infty}\beta_n^{(L)}(\alpha)
=
\infty.
\]
This means that neither extremely rare nor extremely frequent phenotype mutation can make it possible for selection to favor the abundance of cooperation. When $\nu$ is very small, phenotypic variation is generated too slowly for discrimination to create effective assortment. When $\nu$ is very large, phenotypes are mixed too rapidly across the torus, so phenotypic similarity no longer provides a reliable signal of common strategy.
\end{remark}

\begin{corollary}\label{cor:nu_optimal}
For fixed \(\alpha>0\), \(L\ge2\), \(n\ge1\), and \(\mu>0\), the map \(\nu\longmapsto \beta_n^{(L)}(\alpha)\) is not monotone on \((0,\infty)\) and attains at least one global minimum at some interior point \(\nu^*>0\).
\end{corollary}

\begin{proof}
By Theorem \ref{thm:nu_asymptotics}, \(\beta_n^{(L)}(\alpha)\to\infty\) at both endpoints of \((0,\infty)\). Since \(\nu\mapsto\beta_n^{(L)}(\alpha)\) is continuous on $(0,\infty)$ and finite for every finite \(\nu>0\), it attains at least one global minimum at an interior point \(\nu^*>0\). The divergence at both endpoints also rules out monotonicity on $(0,\infty)$.
\end{proof}

\section{\label{AppendixE}Effect of the Strategy Mutation Rate}
In this section, we will study the effect of the strategy mutation rate \(\mu\) on the threshold $\beta_n^{(L)}(\alpha)$.
\begin{lemma}\label{lem:stieltjes_measure}
Define the discrete probability measure \(\rho\) on \([0,\infty)\) by
\[
\rho
=
\frac{1}{L^n}
\sum_{p_1,\ldots,p_n=0}^{L-1}
\left(\prod_{j=1}^n \hatphi(p_j)\right)
\delta_{\lambda_{p_1,\ldots,p_n}},
\]
where \(\lambda_{p_1,\ldots,p_n}=\frac{2\nu}{n}\left(n-\sum_{j=1}^n c_{p_j}\right)\). Then, for every $s>0$, we have
\begin{equation}\label{lemma:stieljes_form}
\LnL\left(\alpha,\frac{s}{2\nu}\right)=\frac{2\nu}{n}\int_{[0,\infty)}\frac{1}{\lambda+s}\,d\rho(\lambda).
\end{equation}
\end{lemma}

\begin{proof}
The proof is a direct consequence of rewriting Eq.~\eqref{eq:closed_form} as
\begin{align*}
\LnL\left(\alpha,\frac{s}{2\nu}\right)&=\frac{1}{L^n}\sum_{p_1,\ldots,p_n=0}^{L-1}\frac{\prod_{j=1}^n \hatphi(p_j)}{n-\sum_{j=1}^n c_{p_j}+\frac{ns}{2\nu}}\\
&=\frac{2\nu}{nL^n}\sum_{p_1,\ldots,p_n=0}^{L-1}\frac{\prod_{j=1}^n \hatphi(p_j)}{\lambda_{p_1,\ldots,p_n}+s}\\
&=\frac{2\nu}{n}\int_{[0,\infty)}\frac{1}{\lambda+s}\,d\rho(\lambda).
\end{align*}
\end{proof}

\begin{theorem}\label{thm:mu}
Fix $\alpha>0$, $L\ge2$, $n\ge1$, and $\nu>0$. Then the map $\mu\longmapsto \beta_n^{(L)}(\alpha)$ is strictly increasing on $(0,\infty)$.
\end{theorem}
\begin{proof}
Evaluating \eqref{lemma:stieljes_form} at
\(s=1,1+2\mu,3+2\mu\) and substituting into
Eq.~\eqref{eq:betaexplicit}, we can write the threshold as
\begin{equation}\label{eq:beta_mu}
\beta_n^{(L)}(\alpha) =\frac{\int_{[0,\infty)}\Xi_1(\lambda,\mu)\,d\rho(\lambda)}{\int_{[0,\infty)}\Xi_2(\lambda,\mu)\,d\rho(\lambda)},
\end{equation}
where
\begin{align}
\Xi_1(\lambda,\mu)&=\frac{(1+2\mu)^2}{\lambda+1}-\frac{(1+\mu)(1+2\mu)}{\lambda+1+2\mu}+\frac{\mu(3+2\mu)}{\lambda+3+2\mu}\nonumber\\
&=\frac{4\mu(1+\mu)\bigl(\lambda^2+(4\mu+3)\lambda+4\mu^2+8\mu+3\bigr)}{(\lambda+1)(\lambda+2\mu+1)(\lambda+2\mu+3)},\nonumber
\\
\Xi_2(\lambda,\mu)
&=
-\frac{1+2\mu}{\lambda+1}
+\frac{(1+\mu)(1+2\mu)}{\lambda+1+2\mu}
+\frac{\mu(3+2\mu)}{\lambda+3+2\mu}\nonumber\\
&=
\frac{
4\lambda\mu(1+\mu)(\lambda+2\mu+2)}
{(\lambda+1)(\lambda+2\mu+1)(\lambda+2\mu+3)}.\label{eq:Fmu}
\end{align}

Since \(\Xi_2(\lambda,\mu)>0\) for \(\lambda>0\), and since
\(\rho\) assigns positive mass to \((0,\infty)\), we have
\[
\int_{[0,\infty)}\Xi_2(\lambda,\mu)\,d\rho(\lambda)>0.
\]
Moreover, \(\rho\) has finite support, so differentiation under the
integrals is immediate. Differentiating \eqref{eq:beta_mu} with respect
to \(\mu\) and symmetrizing the resulting double integral, we obtain
\[
\frac{\partial}{\partial\mu}\beta_n^{(L)}(\alpha)
=
\frac{
\frac12
\iint_{[0,\infty)^2}
K_\mu(\lambda,\eta)\,
d\rho(\lambda)\,d\rho(\eta)}
{
\left(
\int_{[0,\infty)}
\Xi_2(\lambda,\mu)\,d\rho(\lambda)
\right)^2},
\]
where
\begin{align}\label{kernel:eq1}
K_\mu(\lambda,\eta)&=
\Xi_2(\eta,\mu)
\frac{\partial}{\partial\mu}\Xi_1(\lambda,\mu)
+
\Xi_2(\lambda,\mu)
\frac{\partial}{\partial\mu}\Xi_1(\eta,\mu)\nonumber\\
&-
\Xi_1(\eta,\mu)
\frac{\partial}{\partial\mu}\Xi_2(\lambda,\mu)
-
\Xi_1(\lambda,\mu)
\frac{\partial}{\partial\mu}\Xi_2(\eta,\mu).
\end{align} 
Then, to show that $\mu\longmapsto \beta_n^{(L)}(\alpha)$ is strictly increasing on $(0,\infty)$ it is enough to show that 
\[
\iint_{[0,\infty)^2}
K_\mu(\lambda,\eta)\,
d\rho(\lambda)\,d\rho(\eta)>0.
\]

A direct substitution of \eqref{eq:Fmu} into
\eqref{kernel:eq1}, followed by reduction to a common denominator, gives
\begin{equation}\label{eq:Kmu_factor}
K_\mu(\lambda,\eta)
=
\frac{
32\mu^2(1+\mu)^2 P_\mu(\lambda,\eta)}
{D_\mu(\lambda)D_\mu(\eta)},
\end{equation}
where
\begin{align*}
&P_\mu(\lambda,\eta)\\
=&(2\mu+3)^2(\lambda+\eta)^3\lambda\eta+(2\mu+1)^2(2\mu+3)^2(\lambda+\eta)^3\\
&+2(2\mu+3)(\lambda+\eta)^2\lambda^2\eta^2+(\lambda+\eta)\lambda^3\eta^3\\
&+16(1+\mu)^2(4\mu+5)(\lambda+\eta)^2\lambda\eta+2(4\mu+1)\lambda^3\eta^3\\
&+4(1+\mu)(2\mu+1)^2(2\mu+3)^2(\lambda+\eta)^2\\
&+(44\mu^2+80\mu+31)(\lambda+\eta)\lambda^2\eta^2\\
&+(176\mu^4+736\mu^3+1160\mu^2+808\mu+203)(\lambda+\eta)\lambda\eta \\
&+(2\mu+1)^2(2\mu+3)^2(4\mu^2+8\mu+5)(\lambda+\eta)\\
&+4(16\mu^3+44\mu^2+36\mu+7)\lambda^2\eta^2\\
&+2(2\mu+1)(2\mu+3)(16\mu^3+52\mu^2+68\mu+31)\lambda\eta
\end{align*}
and 
\[
D_\mu(x)
=
(x+1)(x+2\mu+1)^2(x+2\mu+3)^2.
\]
Since \(\lambda,\eta\geq0\), we have \(P_\mu(\lambda,\eta)\geq0\), and the inequality is strict if \((\lambda,\eta)\neq(0,0)\). Consequently, $K_\mu(\lambda,\eta)>0$ whenever $(\lambda,\eta)\neq(0,0)$, whereas \(K_\mu(0,0)=0\).
As \(\rho\) is not concentrated at \(0\), where $\rho((0,\infty))>0$, it follows that
\[
\iint_{[0,\infty)^2}
K_\mu(\lambda,\eta)\,
d\rho(\lambda)\,d\rho(\eta)>0,
\]
which completes the proof.
\end{proof}

\begin{corollary}\label{cor:mu_infty_revised}
For fixed $\alpha>0$, $L\ge2$, $n\ge1$, and $\nu>0$, we have
\[
\lim_{\mu\to\infty}\beta_n^{(L)}(\alpha)=\infty.
\]
\end{corollary}

\begin{proof}
 The asymptotic formula in \eqref{eq28} gives
\begin{equation}\label{eq30}
\beta_n^{(L)}(\alpha)
\sim
\frac{2\LnL(\alpha,\frac1{2\nu})}
{\frac{2\nu}{n}-\LnL(\alpha,\frac1{2\nu})}\,\mu,
\end{equation}
as $\mu\to\infty$.
The coefficient of $\mu$ in \eqref{eq30} is strictly positive because
\[
\begin{aligned}
\LnL\left(\alpha,\frac1{2\nu}\right)&=\int_0^\infty [\Phic(\sigma)]^n e^{-n\sigma/(2\nu)}\,d\sigma\\
&<
\int_0^\infty e^{-n\sigma/(2\nu)}\,d\sigma
=
\frac{2\nu}{n}.
\end{aligned}
\]
Therefore, the threshold $\beta_n^{(L)}(\alpha)$ grows linearly in $\mu$ and tends to infinity as $\mu\to\infty$.
\end{proof}

\bibliographystyle{apsrev4-2}
\bibliography{apssamp}

\end{document}